\documentclass[12pt,letterpaper]{amsart}
\usepackage[utf8]{inputenc}
\usepackage[margin=1in,footskip=0.25in]{geometry}
\usepackage{amsmath, amssymb, amsthm, amsfonts}
\usepackage{graphicx}
\usepackage{xcolor}
\usepackage{bm}
\usepackage{hyperref}
\usepackage{subfig}
\usepackage{cite}
\usepackage{mathtools}
\usepackage{cleveref}
\usepackage{commath}
\usepackage[foot]{amsaddr}

\title[Discrete scalar curvature from ORC]{Discrete scalar curvature as a weighted sum of Ollivier-Ricci curvatures}
\author{Abigail Hickok$^1$ and Andrew J. Blumberg$^{1, 2, 3}$}
\address{$^1$Department of Mathematics, Columbia University, New York, NY, USA. $^2$Department of Computer Science, Columbia University, New York, NY, USA. $^3$Irving Institute for Cancer Dynamics, Columbia University, New York, NY, USA.}
\thanks{AH was supported by NSF grant DMS-2303402. AJB was partially supported by ONR grant N00014-22-1-2679.}

\date{\today}

\newcommand{\ratiosumA}{T_N}

\newcommand{\R}{\mathbb{R}}
\newcommand{\Ric}[3]{\text{Ric}_{#1}(#2, #3)}
\newcommand{\Ricc}[2]{\text{Ric}_{#1}(#2, #2)}
\newcommand{\Riccc}[2]{\widehat{\text{Ric}}_{#1}(#2)}
\newcommand{\orc}[2]{\kappa_G(#1, #2)}
\newcommand{\sorc}[1]{\text{SORC}(#1)}
\newcommand{\orsc}[1]{\text{SORC}(#1)}
\newcommand{\rsc}{\text{SRC}}
\newcommand{\fracc}[1]{\frac{1}{#1}}
\newcommand{\E}{\mathbb{E}}
\newcommand{\p}{\mathbb{P}}
\newcommand{\tr}{\text{tr}}
\newcommand{\var}{\text{var}}

\newcommand{\dimn}{n}
\newcommand{\vol}{\text{vol}}
\newcommand{\BM}[2]{B^M({#1}, {#2})}
\newcommand{\Zntrunc}{\overline{Z_{N, k}}}
\newcommand{\Ztrunc}{\overline{Z_k}}
\newcommand{\Snk}{S_{N, k}} % sum of k RVs distributed like Z_{N}
\newcommand{\Snktrunc}{\overline{S_{N, k}}} % sum of k RVs distributed like truncated Z_N
\newcommand{\W}[1]{W_1^{#1}}
\newcommand{\Winf}{W_{\infty}^M}
\newtheorem{theorem}{Theorem}[section]

\theoremstyle{definition}
\newtheorem{remark}[theorem]{Remark}

\newtheorem{lemma}[theorem]{Lemma}
\newtheorem{proposition}[theorem]{Proposition}

\theoremstyle{definition}
\newtheorem{definition}[theorem]{Definition}

\numberwithin{equation}{section}

\begin{document}

\maketitle

\begin{abstract}
We study the relationship between discrete analogues of Ricci and scalar curvature that are defined for point clouds and graphs.  In the discrete setting, Ricci curvature is replaced by Ollivier-Ricci curvature.  Scalar curvature can be computed as the trace of Ricci curvature for a Riemannian manifold; this motivates a new definition of a scalar version of Ollivier-Ricci curvature.  We show that our definition converges to scalar curvature for nearest neighbor graphs obtained by sampling from a manifold.  We also prove some new results about the convergence of Ollivier-Ricci curvature to Ricci curvature.
\end{abstract}

\section{Introduction}

The curvature of a Riemannian manifold $M$ is a fundamental way of characterizing its geometry.  The curvature of a manifold roughly speaking encodes the way that the tangent plane is moving as we move around on the manifold.  Unlike Euclidean space (which is a ``flat'' manifold with zero curvature), manifolds in general may be curved either positively or negatively. Curvature affects properties such as whether parallel geodesics (locally-shortest paths) curve either towards or away from each other, and whether geodesic balls are smaller or larger than their Euclidean analogs. It also constrains the global topology of a manifold.

In the case of a curve, the curvature is essentially encoded in the second derivative.  But in general, the curvature is a much more complicated quantity.  There are several distinct but intimately related notions of curvature; we consider two of them in this paper.  Arguably the simplest notion of curvature is the {\em scalar curvature}, which is a function $S\colon M \to \R$ that quantifies the local curvature at each point $x \in M$.  Scalar curvature is an isometry-invariant, and manifolds with positive scalar curvature have been extensively studied; this condition puts constraints on the possible topology of the manifold, and moreover arises naturally in the theory of general relativity.  A richer notion of curvature is captured by the {\em Ricci curvature} tensor, a linear transformation $\text{Ric}_x\colon T_xM \to T_xM$ on the tangent space at $x$.  Ricci curvature controls the local deformations of shapes along geodesic paths on the manifold, and plays an essential role in our understanding of the geometry of $3$-manifolds (as indicated by Perelman's proof of Thurston's geometrization conjecture).
Scalar and Ricci curvature measure curvature in different ways, but they are related by the fact that scalar curvature is the trace of Ricci curvature.

In recent years, there has been growing interest in discrete analogues of curvature for graphs. There are several definitions of ``discrete Ricci curvature'' that take the form of functions $\kappa\colon E \to \R$ on the edges of a graph (e.g., see \cite{ollivier_ricci, forman_ricci, resistance_curvature}). 
We focus on Ollivier-Ricci curvature (ORC), a version of discrete Ricci curvature that is based on optimal transport. For a weighted graph $G$, the ORC of an edge $(x, y)$ is
\begin{equation*}
    \kappa_G(x, y) = 1 - \frac{W_1(\mu_x, \mu_y)}{w(x, y)}\,,
\end{equation*}
where $\mu_x, \mu_y$ are uniform probability measures on neighbors of nodes $x, y$ respectively, and $W_1$ denotes the $1$-Wasserstein metric. In \Cref{fig:orc_examples}, we show three prototypical examples of edges with negative, zero, and positive ORC. Edges that are present in highly interconnected ``communities'' tend to have positive ORC, while ``local bottlenecks'' tend to have negative ORC.

\begin{figure}
\centering
\subfloat[Negative]{\includegraphics[width =.25\textwidth]{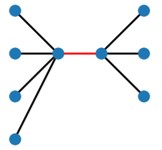}}
\hspace{8mm}
\subfloat[Zero]{\includegraphics[width =.25\textwidth]{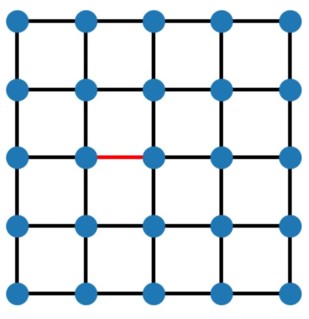}}
\hspace{8mm}
\subfloat[Positive]{\includegraphics[width =.25\textwidth]{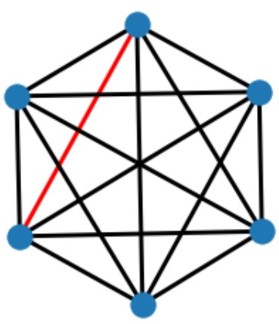}}
\caption{Edges (in red) with (A) negative ORC, (B) zero ORC, and (C) positive ORC. We set each edge weight to $1$ and calculate $\kappa_G(x, y)$ using the formulation in \Cref{def:orc}.}\label{fig:orc_examples}
\end{figure}

ORC has been used in numerous applications. In \cite{sia_ollivier-ricci_2019, ni_community_2019}, ORC was applied to the problem of community detection \cite{porter_communities_2009}.  Recently, Saidi et al. used ORC to 
improve low-dimensional embeddings of high-dimensional point cloud data \cite{saidi_recovering_2024, embedor}. Within machine learning, it has been used to mitigate over-smoothing and over-squashing \cite{nguyen_2023} in graph neural networks, to encode structural properties \cite{ye_curvature_2020, fesser_2024}, and to evaluate graph generative models \cite{rieck_2023}.

From a theoretical perspective, Ollivier-Ricci curvature is particularly appealing since it is the only form of discrete Ricci curvature with formal guarantees about its convergence to traditional Ricci curvature. When ORC was first introduced by Ollivier, it was defined not only on graphs, but also on any metric space $M$ equipped with a Markov kernel $\{\mu_x\}_{x \in M}$. For example, in the manifold version of ORC, $M$ is a Riemannian manifold and $\mu_x$ is the uniform probability measure on a geodesic ball centered at $x$. Ollivier proved that his manifold version of ORC (see \Cref{def:orc_manifold}) converges to the Ricci curvature in an appropriate limit. Van der Hoorn et al. \cite{ORC_convergence_hoorn} extended this result to graph ORC on random geometric graphs, where nodes are sampled uniformly at random from a manifold $M$ and edges connect points within a small connectivity threshold $\epsilon$ of each other. They proved that for pairs of points at prescribed distance $\delta \to 0$, a scaled version of ORC converges in probability to a constant multiple of the Ricci curvature of $M$. Trillos and Weber \cite{trillos_weber} strengthened these results by proving non-asymptotic bounds on the rate of convergence. These theoretical results establish that one can regard ORC as a sensible way to approximate the Ricci curvature of an underlying manifold from samples.

In many applications, it is desirable to have a definition of discrete scalar curvature $S\colon V \to \R$ on the nodes $V$ of a graph, so that one can study the node-wise properties of the graph. Given the relationship between Ricci and scalar curvature, it is natural to ask whether one can use discrete Ricci curvature to define discrete scalar curvature. Sandhu et al. \cite{curvature_cancer} and Sreejith et al. \cite{other_scalar} did just this by defining the discrete scalar curvature of a node $x$ to be the sum (or mean, depending on which version one uses) of the discrete Ricci curvatures of the incident edges; the former used Forman-Ricci curvature, and the latter used Ollivier-Ricci curvature. However, it is not known whether either of these definitions converge to the scalar curvature of the underlying manifold.

\subsection*{Contributions}

The purpose of this paper is to introduce a new definition of discrete scalar curvature and study its relationship to the scalar curvature of manifolds. Given a weighted graph $G$ with edge weights $w(x, y)$, we define scalar Ollivier-Ricci curvature
\begin{equation}\label{eq:orsc}
    \sorc{x} := \fracc{\deg(x)}\sum_{y \text{ adjacent to } x} w(x, y)^2\orc{x}{y}
\end{equation}
For an unweighted graph, one can calculate $\orsc{x}$ by setting $w(x, y) = 1$ for each edge. We use the version of Ollivier-Ricci curvature $\kappa_G(x, y)$ that is defined in \Cref{def:orc}. 

This definition uses a weighted sum of the discrete Ricci curvatures of the incident edges, in contrast to the definitions of \cite{curvature_cancer} and \cite{other_scalar}. In particular, when the nodes are points on a manifold and the graph weights $w(x, y)$ are their pairwise geodesic distances $d_M(x, y)$, our weighted sum puts higher weight on pairs of points that are farther away from each other. 

There are two motivations for this definition. The first is that this weighting scheme is the correct scaling factor in order to obtain convergence to a constant multiple of scalar curvature, as we explain. The second is that we show that the relationship between $\kappa_G(x, y)$ and Ricci curvature tends to deteriorate as $d_M(x, y) \to 0$; see~\Cref{prop:orc_convergence1}.

Our main theorem says that a scaled version of $\orsc{x}$ converges to scalar curvature $S(x)$ in probability. We assume that the nodes of our graph are sampled from a compact Riemannian manifold $M$ of dimension $n$. Details about our random geometric graphs are given in \Cref{sec:setting}.
\begin{theorem}[\Cref{thm:main}]
As the number of nodes $N \to \infty$ and the connection threshold $\epsilon_N \to 0$ (at an appropriate rate) in a random geometric graph $G_N$, we have
    \begin{equation*}
        \Big|\frac{2(n+2)^2}{\epsilon_N^4} \sorc{x_N} - S(x_N)\Big| \xrightarrow[]{\p} 0
    \end{equation*}
for a random node $x_N \in G_N$.
\end{theorem}

As part of establishing the main theorem, we prove two new results that clarify the convergence of Ollivier-Ricci curvature to Ricci curvature. In \Cref{prop:orc_convergence1}, we prove a non-asymptotic upper bound on the error between Ricci curvature and scaled Ollivier-Ricci curvature for all edges in a random geometric graph.
\begin{proposition}[\Cref{prop:orc_convergence1}]
For a constant $\theta > 0$ that is sufficiently small, there are constants $C_{1, \theta}, C_{2, \theta}$ such that with probability at least $1 - C_{1, \theta}N^{-\theta}$,
    \begin{equation*}
        |\epsilon_N^{-2} \orc{x}{y} - \Ricc{x}{v_{xy}}| \leq C_{2, \theta} \Big( \epsilon_N +  \frac{\log(N)^{p_n}}{d_M(x, y) N^{1/n - 2\alpha}}\Big)
    \end{equation*}
    for all edges $(x, y) \in E(G_N)$, where $v_{xy} = \log_x(y)/\norm{\log_x(y)}$.
\end{proposition}
\noindent This improves on the analogous upper bounds in \cite{trillos_weber}, which hold only for edges that are sufficiently short and sufficiently long. Similarly, the results in \cite{ORC_convergence_hoorn} hold only for a given edge with prescribed distance $\delta_N$; moreover, they only proved convergence in mean.

Additionally, we prove in \Cref{prop:orc_convergence2} that the scaled Ollivier-Ricci curvature of a random edge converges in probability to Ricci curvature. Again, we do not restrict to edges of certain lengths, unlike in \cite{ORC_convergence_hoorn, trillos_weber}.  
\begin{proposition}\label{prop:orc_convergence2}
    Let $x_N$ be a node chosen uniformly at random from $G_N$ and let $y_N$ be an adjacent node chosen uniformly at random. As the number of nodes $N \to \infty$,
    \begin{equation*}
        \epsilon_N^{-2} \kappa_G(x_N, y_N) \xrightarrow[]{\p} \Ricc{x_N}{v_N} \,,
    \end{equation*}
    where $v_N = \log_{x_N}(y_N)/\norm{\log_{x_N}(y_N)}$.
\end{proposition}

Numerical experiments on samples from manifolds of known scalar curvature validate our theoretical results and give a sense of the rate of convergence.

\subsection*{Acknowledgements}
We thank Dylan Altschuler and Milind Hegde for helpful discussions.

\section{Background and Assumptions}

The purpose of this section is to give a terse review of the key definitions used in the paper.

\subsection{Riemannian geometry and curvature} 

Let $M$ be a compact $\dimn$-dimensional Riemannian manifold, where $\dimn \geq 2$. The Riemannian metric allows us to define the norm $\norm{v}$ of any vector $v$ in the tangent space $T_xM$, for any point $x$ in $M$. This in turn allows us to define the \emph{geodesic distance} $d_M(x, y)$ between every pair $x, y$ of points on the manifold.

A closed geodesic ball of radius $\epsilon$, centered at $x \in M$, is defined as
\begin{equation*}
    \BM{x}{\epsilon} := \{y \in M \mid d_M(x, y) \leq \epsilon \}\,.
\end{equation*}
\noindent We denote its volume by $\vol(\BM{x}{\epsilon})$. We use $v_n$ to denote the volume of an $n$-dimensional unit Euclidean ball.

The Ricci curvature at $x$ is a bilinear map
\begin{equation*}
    \text{Ric}_x : T_xM \times T_xM \to \R\,.
\end{equation*}
Equivalently, we can associate $\text{Ric}_x$ with a linear transformation 
\[
\Riccc{x}{v}: T_xM \to T_xM \qquad
    \Riccc{x}{v} = \sum_i \Ric{x}{v}{e_i}e_i \,,
\]
where $\{e_i\}_{i=1}^\dimn$ is an orthonormal basis for the vector space $T_xM$.  Observe that the two definitions are related by the formula $\Riccc{x}{v}\cdot w = \Ric{x}{v}{w}$ for all $v, w \in T_xM$. 

The scalar curvature at $x$ is the trace of the Ricci curvature:
\begin{equation}\label{eq:trace}
    S(x) = \tr(\widehat{\text{Ric}}_x) = \sum_i \Riccc{x}{e_i} \cdot e_i \,.
\end{equation}

\subsection{Random geometric graphs}\label{sec:setting}

\begin{definition}A \emph{random geometric graph} (RGG) on a manifold $M$ is a weighted graph $G$ where the nodes are labeled by points in $M$ and an edge $(x, y)$ exists if and only $d_M(x, y)$ is less than a \emph{connectivity threshold} $\epsilon$. The weight of an edge $(x, y)$ is the geodesic distance $d_M(x, y)$.
\end{definition}

We denote the edge set of a graph $G$ by $E(G)$. When two nodes $x, y$ are adjacent, we write $x \sim y$.

Throughout this paper, let $G_N$ be an RGG with $N$ nodes and connection threshold $\epsilon_N$, whose vertices are sampled uniformly at random from $M$, where $\epsilon_N = N^{-\alpha}$ for some $\alpha \in (0, \fracc{6n})$. Let $\{x_N\}$ be any sequence of nodes $x_N \in G_N$.

The shortest-path distance $d_G(x, y)$ between two nodes $x, y \in G$ is the (weighted) length of a shortest path between $x$ and $y$. 

\subsection{Ollivier-Ricci curvature}
It is helpful to first define Ollivier-Ricci curvature for pairs of points on a manifold, and then to extend the definition to a discrete notion of curvature for edges in a graph. See \cite{ollivier_ricci} for more details.
\begin{definition}\label{def:orc_manifold}[Ollivier-Ricci curvature on a manifold] Let $M$ be a manifold, and let $\epsilon > 0$. For every point $x \in M$, let $\mu_x^M$ be the uniform probability measure on the ball $B^M(x, \epsilon)$. The \emph{manifold Ollivier-Ricci curvature} of a pair $x, y$ is \[\kappa_M(x, y) = 1 - \frac{W_1^M(\mu_x^M, \mu_y^M)}{d_M(x, y)}\,,\] where the $1$-Wasserstein distance $W_1^M$ is calculated with respect to geodesic distance $d_M$.
\end{definition}

Ollivier \cite{ollivier_ricci} proved the following theorem, which says that manifold Ollivier-Ricci curvature converges to Ricci curvature as $\epsilon \to 0$ and $d_M(x, y) \to 0$.

\begin{theorem}\label{thm:KM_convergence}
    \begin{equation*}
        \epsilon^{-2}\Big\vert \kappa_M(x, y) - \frac{\Ricc{x}{v_{xy}}}{2(n+2)} \Big| = \mathcal{O}\Big(\epsilon + d_M(x, y)\Big)\,,
    \end{equation*}
    where $v_{xy} = \log_x(y) \in T_xM$.
\end{theorem}

\begin{definition}[Ollivier-Ricci curvature on a graph]\label{def:orc}
    Let $G$ be a weighted graph with edge weights $w(x, y)$. For each node $x \in G$, let $\mu_x^G$ be the uniform probability measure on its $1$-hop neighborhood, excluding $x$ itself. The \emph{Ollivier-Ricci curvature} of an edge $(x, y)$ is \[\kappa_G(x, y) = 1 - \frac{W_1^G(\mu_x^G, \mu_y^G)}{w(x, y)},\] where the $1$-Wasserstein distance $W_1^G$ is calculated with respect to the shortest-path metric on the graph.
\end{definition}
\noindent If $G$ is unweighted, one can define $\kappa_G(x, y)$ by taking the weight of each edge to be $1$. In our paper, the graphs that we consider are random geometric graphs, so the edge weights are the geodesic distances $d_M(x, y)$.
\begin{remark}
    Ollivier-Ricci curvature is defined much more generally in \cite{ollivier_ricci}. Given any Markov kernel $\{\mu_x^G\}$, the Ollivier-Ricci curvature of edge $(x, y)$ is $\kappa_G(x, y) = 1 - \frac{W_1(\mu_x^G, \mu_y^G)}{w(x, y)}$. However, for the rest of this paper, all references to Ollivier-Ricci curvature $\kappa_G(x, y)$ are specifically the formulation in \Cref{def:orc}.
\end{remark}

Van der Hoorn et al. proved that Ollivier-Ricci curvature converges to Ricci curvature on random geometric graphs in an appropriate limit, for edges with prescribed distance $\delta \to 0$.
\begin{theorem}[\cite{ORC_convergence_hoorn}]\label{thm:hoorn_convergence}
    Let $G$ be a random geometric graph (with connection threshold $\epsilon_N$) whose vertices are drawn from a Poisson point process on $M$ with rate $N$. The edge weights are given by geodesic distance. Let $x$ be a fixed point in $M$ and let $v \in T_xM$ be a fixed unit vector. Let $y_N$ be a point that is geodesic-distance $\delta_N$ from $x$ such that $v = \log_x(y_N)/\norm{\log(y_N)}$. Add $x$ and $y_N$ to $G$. If $\epsilon_N \sim N^{-\alpha}$ and $\delta_N \sim N^{-\beta}$, where
    \begin{equation*}
        0 < \beta \leq \alpha \,, \qquad \alpha + 2 \beta < \fracc{n}\,,
    \end{equation*}
    then 
    \begin{equation*}
        \lim_{N \to \infty} \E\Big[ \Big\vert \frac{\orc{x}{y_N}}{\delta_N^2} - \frac{\Ricc{x}{v}}{2(n+2)} \Big\vert \Big] = 0\,,
    \end{equation*}
    where $\orc{x}{y_N}$ is calculated with $\delta_N$-radius graph balls.
\end{theorem}
\noindent They also proved convergence if one does not have access to the pairwise geodesic distances, but only in the case $\dim(M) = 2$.
\begin{theorem}[\cite{ORC_convergence_hoorn}]
    Suppose that $\dim(M) = 2$ and the weight of every edge is $\epsilon_N$. Then 
    \begin{equation*}
        \lim_{N \to \infty} \E\Big[ \Big\vert \frac{\orc{x}{y_N}}{\delta_N^2} - \frac{\Ricc{x}{v}}{2(n+2)} \Big\vert \Big] = 0\,,
    \end{equation*}
    assuming that
    \begin{equation*}
        0 < \beta < \fracc{9}\,, \qquad 2 \beta < \alpha < \frac{1 - 3\beta}{2}\,.
    \end{equation*}
\end{theorem}

More recently, Trillos and Weber \cite{trillos_weber} strengthened the results above by proving non-asymptotic upper bounds on the error, for all edges that are sufficiently short and sufficiently long. They assume that the nodes in the graph are points on a manifold embedded in Euclidean space, and that one has access to the Euclidean coordinates of the points (i.e., the nodes are a point cloud).
\begin{theorem}[\cite{trillos_weber}]
    Assume that one has access to the pairwise geodesic distances $d_M(x, y)$. Define $p_n = 3/4$ if $n = 2$, and $p_n = 1/m$ if $n \geq 3$. Then one can construct a graph metric $d_G(x, y)$ such that for every $s > 1$, there is a constant $C$ such that with probability at least $1 - CN^{-s}$,
    \begin{equation*}
        \Big\vert \frac{\kappa_G(x, y)}{\epsilon^2} - \frac{\Ricc{x}{v}}{2(n+2)}\Big\vert \leq C \Big( \epsilon + \frac{\log(N)^{p_n}}{N^{1/n}\epsilon^3}\Big)\,,
    \end{equation*}
    for all $x, y$ satisfying $2c_0 \epsilon \leq d_M(x, y) \leq \fracc{2} c_1 \epsilon$, where $c_0$ and $c_1$ are fixed but sufficiently small and large, respectively.\footnote{The constants in the bound depend on $c_0$ and $c_1$, so one cannot simply take $c_0 \to 0$ or $c_1 \to \infty$ in order to consider shorter or longer edges.} The Ollivier-Ricci curvature $\kappa_G(x, y)$ is calculated with respect to the graph distance $d_G$.\footnote{They calculate $\kappa_G(x, y) = 1 - \frac{W_1^G(\mu_x^G, \mu_y^G)}{d_G(x, y)}$, where $W_1^G$ is calculated with respect to their constructed graph metric $d_G$, which is not necessarily the shortest-path metric.}
\end{theorem}
\noindent 

\section{The mean Ricci curvature at a node converges to scalar curvature}\label{sec:mean_ricci_curvature}
We know from Riemannian geometry that scalar curvature $S(x)$ is the trace of Ricci curvature at node $x$. (See \Cref{eq:trace}.) This leads us to define ``discrete scalar Ricci curvature''
\begin{equation*}
    \rsc(x):= \fracc{\deg(x)} \sum_{y \text{ adj. to } x} \Riccc{x}{\log_x y} \cdot \log_x y\,,
\end{equation*}
for every vertex $x$. Unfortunately, the edges incident to a node $x$ do not necessarily form an orthonormal basis for the tangent space $T_x(M)$, so $\rsc(x)$ is not equal to the trace of the Ricci curvature at $x$. 

In this section, we use a modification of ``Hutchinson's trick'' \cite{hutchinson_stochastic_1990}---a method for estimating the trace of a linear operator---to prove that $\rsc(x)$ converges to $\tr(\text{Ric}_x) = S(x)$ as $N \to \infty$ (\Cref{lem:rsc_convergence}). Hutchinson's trick says that if $A$ is a symmetric matrix and $u_1, \ldots, u_n$ are i.i.d. random variables with mean $0$ and variance $\sigma^2$, then $\E[\vec{u}^TA\vec{u}] = \sigma^2 \tr(A)$. In this section, the linear transformation $\widehat{\text{Ric}}_x$ plays the role of $A$, and $u_i = \log_x(z) \cdot e_i$, where $\{e_i\}$ is a fixed choice of orthonormal basis and $z$ is a point sampled uniformly at random from the ball $B^M(x, \epsilon)$. However, our random variables $u_i$ do not satisfy Hutchinson's assumptions. Below, we modify Hutchinson's trick so that it works in our setting.
\begin{lemma}\label{lem:hutchinson}
    Suppose that $A$ is a symmetric matrix and $u \in \R^n$ is a random vector (not necessarily sampled uniformly at random) in $B^E(0, \epsilon)$  such that
    \begin{align*}
        \E[u_i^2] &= \sigma^2 + \mathcal{O}(\epsilon^4) \qquad \text{for all } i\,,\\
        \E[u_iu_j] &= \mathcal{O}(\epsilon^4) \qquad \text{for all } i \neq j
    \end{align*}
    for some $\sigma \in \R$ as $\epsilon \to 0$. Then, as $\epsilon \to 0$,
    \begin{align*}
        \E[u^T Au] &= \sigma^2 \tr(A) + \mathcal{O}(\epsilon^4) \cdot \norm{A}_{1, 1}\,,\\
        \var(u^TAu) &= \mathcal{O}(\epsilon^4) \cdot \norm{A}_{1, 1}^2\,.
    \end{align*}
\end{lemma}
\begin{proof}
    We begin by calculating $\E[u^T Au]$.
    \begin{align*}
        \E[u^TAu] &= \sum_{ij} A_{ij} \E[u_iu_j] \\
        &= \sum_{ij} A_{ij} \Big(\delta_{ij} \sigma^2 + \mathcal{O}(\epsilon^4)\Big)\\
        &= \sigma^2 \tr(A) + \mathcal{O}(\epsilon^4) \cdot \norm{A}_{1, 1}\,.
    \end{align*}
    Now we consider $\var(u^TAu)$:
    \begin{align*}
        \var(u^TAu) &= \E[(u^TAu)^2] - \E[u^TAu]^2 \\
        &\leq \E[(u^TAu)^2] \\
        &= \sum_{i, j, r, s} A_{ij}A_{rs} \E[u_iu_ju_ru_s]\,.
    \end{align*}
    We have $\E[u_iu_j u_ru_s] \leq \epsilon^4$ because $\norm{u} \leq \epsilon$. Therefore,
    \begin{equation*}
        \var(u^TAu) \leq \mathcal{O}(\epsilon^4) \norm{A}_{1, 1}^2\,.
    \end{equation*}
\end{proof}

Let $x \in M$ and let $z$ be a point sampled uniformly at random from $B^M(x, \epsilon)$. Assume that $\epsilon$ is smaller than the injectivity radius of $M$. Let $u = \log_x z \in T_xM$, and let $c_M(x, \epsilon) = v_n\epsilon^n/\vol(\BM{x}{\epsilon})$. We fix an orthonormal basis $\{e_i\}$ for $T_xM$ and define $u_i := (\log_x z) \cdot e_i$.

\begin{lemma}\label{lem:u_var}
     For sufficiently small $\epsilon$,
    \begin{equation*}
        \E[u_i^2] = c_M(x, \epsilon) \frac{\epsilon^2}{n+2} + \mathcal{O}(\epsilon^4)\,,
    \end{equation*}
    where $c_M(x, \epsilon) \to 1$ as $\epsilon \to 0$ (uniformly in $x$ as long as $M$ is compact).
\end{lemma}
\begin{proof}
The expectation is
\begin{equation*}
    \E[u_i^2] = \fracc{\vol(\BM{x}{\epsilon}} \int_{\BM{x}{\epsilon}} ((\log_x z)\cdot e_i)^2 \,dV(z)\,.
\end{equation*}
    For sufficiently small $\epsilon$, the exponential map $\exp_x: \BM{x}{\epsilon} \to B^E(0, \epsilon) \subseteq T_xM$ and its inverse $\log_x :B^E(0, \epsilon) \to \BM{x}{\epsilon}$ are diffeomorphisms. In the geodesic normal coordinates $u_1, \ldots, u_n$, we have
    \begin{equation*}
        \int_{\BM{x}{\epsilon}} ((\log_x z)\cdot e_i)^2 dV(z) = \int_{B^E(0, \epsilon)} u_i^2 \sqrt{\det(g)} \,du_1\cdots du_n\,.
    \end{equation*}
    It is standard that
    \begin{equation}\label{eq:det_g}
        \sqrt{\det(g)} = 1 - \fracc{6}\Ric{x}{e_k}{e_\ell} u_k u_\ell + \mathcal{O}(\norm{u}^3)\,.
    \end{equation}
    Therefore
    \begin{align*}
        \E[u_i^2] &= \fracc{\vol(\BM{x}{\epsilon}} \Bigg[ \int_{B^E(0, \epsilon)} u_i^2 du_1 \cdots du_n - \fracc{6}\Ric{x}{e_k}{e_{\ell}} \int_{B^E(0, \epsilon)} u_i^2 u_k u_\ell \,du_1 \cdots u_n \\
        &\qquad\qquad\qquad\qquad\qquad + \int_{B^E(0, \epsilon)} u_i^2 \mathcal{O}(\norm{u}^3) \,du_1 \cdots u_n \Bigg]\,.
    \end{align*}
    Because $\norm{u} \leq \epsilon$, we have $|u_i^2 u_k u_\ell| = \mathcal{O}(\epsilon^4)$ and $|u_i^2 \mathcal{O}(\norm{u}^3)| = \mathcal{O}(\epsilon^5)$, so
    \begin{align*}
        \int_{B^E(0, \epsilon)} u_i^2 u_k u_\ell \,du_1 \cdots u_n &= \mathcal{O}(\epsilon^4) \cdot \vol(B^E(0, \epsilon)) = \mathcal{O}(\epsilon^{n+4})\\
        \int_{B^E(0, \epsilon)} u_i^2 \mathcal{O}(\norm{u}^3) \,du_1 \cdots u_n &= \mathcal{O}(\epsilon^5) \cdot \vol(B^E(0, \epsilon)) = \mathcal{O}(\epsilon^{n+5})\,.
    \end{align*}
    Therefore
    \begin{equation}\label{eq:E_usquared}
        \E[u_i^2] = \fracc{\vol(\BM{x}{\epsilon}} \Bigg(\int_{B^E(0, \epsilon)} u_i^2 \,du_1 \cdots du_n\Bigg) + \mathcal{O}(\epsilon^4)
    \end{equation}
    because $\vol(\BM{x}{\epsilon}) = v_n \epsilon^n + \mathcal{O}(\epsilon^{n+2})$. 
    Let $\mathcal{I}_i := \int_{B^E(0, \epsilon)} u_i^2 \,du_1 \cdots du_n$. The integral $\mathcal{I}_i$ does not depend on $i$, so
    \begin{equation*}
        n \mathcal{I}_i = \sum_{j=1}^n \mathcal{I}_j = \int_{B^E(0, \epsilon)} \norm{u}^2 du_1 \cdots du_n = \int_0^\epsilon r^2 \cdot S_r \,dr\,,
    \end{equation*}
    where $S_r$ is the surface area of an $(n-1)$-dimensional sphere of radius $r$. Therefore,
    \begin{equation*}
        \mathcal{I}_i = \fracc{n} \int_0^\epsilon n v_n r^{n+1}\,dr = \frac{v_n}{n+2}\epsilon^{n+2}\,.
    \end{equation*}
    By \Cref{eq:E_usquared},
    \begin{equation*}
        \E[u_i^2] = c_M(x, \epsilon) \cdot \frac{\epsilon^2}{n+2} + \mathcal{O}(\epsilon^4)\,.
    \end{equation*}
    The last statement of the lemma follows because 
\begin{equation*}
    c_M(x, \epsilon) = \frac{v_n\epsilon^n}{\vol(\BM{x}{\epsilon})} = \frac{v_n\epsilon^n}{v_n \epsilon^n + \mathcal{O}(\epsilon^{n+2})}\,.
\end{equation*}
\end{proof}

\begin{lemma}\label{lem:u_cov}
    For $i \neq j$,
    \begin{equation*}
        \E[u_iu_j] = \mathcal{O}(\epsilon^4)\,.
    \end{equation*}
\end{lemma}
\begin{proof}
    The expectation is
    \begin{equation*}
        \E[u_iu_j] = \fracc{\vol(\BM{x}{\epsilon}} \int_{\BM{x}{\epsilon}} ((\log_x z) \cdot e_i) ((\log_x z) \cdot e_j) \, dV(z) \,.
    \end{equation*}
    In the geodesic normal coordinates $u_1, \ldots, u_n$, we have
    \begin{align*}
        \int_{\BM{x}{\epsilon}} ((\log_x z) \cdot e_i) ((\log_x z) \cdot e_j)\, dV(z) &= \int_{B^E(0, \epsilon)} u_i u_j \sqrt{\det(g)} \,du_1 \cdots du_n \\
        &= \int_{B^E(0, \epsilon)} u_i u_j \,du_1 \cdots du_n + \mathcal{O}(\epsilon^{n+4})
    \end{align*}
    because $\norm{u} \leq \epsilon$. Because $u_i$ and $u_j$ are odd functions,
    \begin{equation*}
        \int_{B^E(0, \epsilon)} u_i u_j \, du_1 \cdots du_n = 0\,,
    \end{equation*}
    so
    \begin{equation*}
        \E[u_i u_j] = \frac{\mathcal{O}(\epsilon^{n+4})}{\vol(\BM{x}{\epsilon})} = \mathcal{O}(\epsilon^4)\,.
    \end{equation*}
\end{proof}

\begin{lemma}\label{lem:degree_large}
    Let $\{x_N\}$ be a sequence of random nodes, with $x_N \in G_N$. Let $k_N = \deg(x_N)$. Then \[\lim_{N \to 0} \p[k_N \leq k] = 0\] for any $k \in \mathbb{Z}$.
\end{lemma}
\begin{proof}
    The degree $k_N$ has a binomial distribution $B(N-1, p_N)$, where \[p_N = \vol(\BM{x}{\epsilon_N})/\vol(M).\] Using the fact that $\vol(\BM{x}{\epsilon_N}) = v_n \epsilon_N^n + \mathcal{O}(\epsilon_N^{n+2})$ as $N \to \infty$, we see that
    \begin{equation*}
        \E[k_N] = (N-1)p_N \geq \frac{v_N}{2\,\vol(M)} \cdot N \epsilon_N^n = N^{1 - n\alpha}
    \end{equation*}
    for sufficiently large $N$. Therefore, $\E[k_N] \to \infty$ because we have stipulated that $1 - n\alpha > 0$.
    
    By applying a Chernoff bound, one obtains $\p[k_N \leq (1 - \delta)\E[k_N]] \leq \exp(-\delta^2\E[k_N]/2)$ for any $0 \leq \delta < 1$ (see e.g., \cite{tsun_chernoff_2020}). Substituting $\delta = 1 - k/\E[k_N]$ yields
    \begin{equation*}
        \p[k_N \leq k] \leq \exp\Big(-\E[k_N]\Big(1 - k/\E[k_N]\Big)^2/2\Big)
    \end{equation*}
    when $N$ is sufficiently large such that $k < \E[k_N]$. Therefore, $\lim_{N \to \infty}\p[k_N \leq k] = 0$ because $\E[k_N] \to \infty$.
\end{proof}

\begin{lemma}\label{lem:rsc_convergence}
As $N \to \infty$, we have
    \begin{equation*}
        \E\Big[ \frac{n+2}{\epsilon_N^2} \rsc(x_N) \Big]  - S(x_N) \to 0
    \end{equation*}
and $\frac{n+2}{\epsilon_N^2}\rsc(x_N) - S(x_N) \to 0$ in probability.
\end{lemma}
\begin{proof}
If $\deg(x_N) = 0$, then $\rsc(x_N) = 0$. If $\deg(x_N) = k > 0$, the random variable $\rsc(x_N)$ is a sample mean of $\{\Riccc{x_N}{\log_{x_N} z} \cdot \log_{x_N} z_i\}_{i=1}^k$, where ${z_i}$ are points sampled uniformly at random from $\BM{x}{\epsilon_N}$, so
\begin{align}
    \E\Big[ \frac{n+2}{\epsilon_N^2} \rsc(x_N) \mid \deg(x_N) > 0 \Big] &= \frac{n+2}{\epsilon_N^2} \cdot \E\Big[ \Riccc{x_N}{\log_{x_N} z_1} \cdot \log_{x_N} z_1\Big]\,,\label{eq:rsc_exp_1} \\
    \var\Big[\frac{n+2}{\epsilon_N^2} \rsc(x_N) \mid \deg(x_N) = k \Big] &= \frac{(n+2)^2}{\epsilon_N^4} \cdot \frac{\var(\Riccc{x_N}{\log_{x_N} z_1} \cdot \log_{x_N} z_1)}{k} \label{eq:rsc_var_1}\,,
\end{align}
By \Cref{lem:hutchinson}, \Cref{lem:u_var}, and \Cref{lem:u_cov},
\begin{align}
     \E\Big[\Riccc{x_N}{\log_{x_N} z_1} \cdot \log_{x_N} z_1\Big] &= c_M(x_N, \epsilon_N) \cdot \frac{\epsilon_N^2}{n+2} \cdot S(x_N) + \mathcal{O}(\epsilon_N^2)\,, \label{eq:rsc_exp_2}\\
     \var\Big(\Riccc{x}{\log_{x_N} z_1} \cdot \log_{x_N} z_1\Big) &= \mathcal{O}(\epsilon_N^4) \label{eq:rsc_var_2}\,,
\end{align}
where $c_M(x_N, \epsilon_N) \to 1$ as $\epsilon_N \to 0$. Together, \Cref{eq:rsc_exp_1} and \Cref{eq:rsc_exp_2} imply that
\begin{align*}
    \Big\vert \E\Big[\frac{n+2}{\epsilon_N^2} \rsc(x_N)\Big] - S(x_N)\Big\vert &= \Big\vert\E\Big[\Riccc{x}{\log_{x_N} z_1} \cdot \log_{x_N} z_1\Big] \cdot \p[\deg(x_N) > 0] - S(x_N)\Big\vert \\
    &\leq S(x_N)\big\vert c_M(x_N, \epsilon_N)\,\p[\deg(x_N) > 0] -1 \big\vert
\end{align*}
as $\epsilon_N \to 0$, where $c_M(x_n, \epsilon_N) \to 1$ as $N \to \infty$. By \Cref{lem:degree_large}, $\p[\deg(x_N) > 0] \to 1$ as $N \to \infty$, so the right-hand side converges to $0$, which proves convergence of the expectation as $N \to \infty$. 

Let $\delta > 0$ and $\xi > 0$. We wish to show that
\begin{equation*}
    \p\Big[ \Big\vert \frac{n+2}{\epsilon_N^2} \rsc(x_N) - S(x_N) \Big\vert > \delta \Big] < \xi
\end{equation*}
for sufficiently large $N$. In the next steps, we condition on the value of $\deg(x_N)$. \Cref{eq:rsc_exp_1} and \Cref{eq:rsc_exp_2} imply that
\begin{equation*}
    \Big\vert \E\Big[ \frac{n+2}{\epsilon_N^2} \cdot \rsc(x_N) \mid \deg(x_N) = k \Big] - S(x_N) \Big\vert < \delta/2 
\end{equation*}
for sufficiently large $N$ and all $k > 0$. Therefore,
\begin{align*}
    &\p\Bigg[ \Big| \frac{n+2}{\epsilon_N^2} \cdot \rsc(x_N) - S(x_N) \Big| > \delta \mid \deg(x_N) = k \Bigg] \\
    &\qquad\qquad \leq \p\Bigg[ \Big| \frac{n+2}{\epsilon_N^2} \cdot \rsc(x_N) - \E\Big[ \frac{n+2}{\epsilon_N^2} \rsc(x_N) \mid \deg(x_N) = k \Big] \Big| > \delta/2 \mid \deg(x_N) = k \Bigg]\,.
\end{align*}
By apply Chebyschev's inequality to the right-hand side, we obtain
\begin{align*}
     \p\Bigg[ \Big| \frac{n+2}{\epsilon_N^2} \cdot \rsc(x_N) - S(x_N) \Big| > \delta \mid \deg(x_N) = k \Bigg] \leq \frac{4}{\delta^2} \cdot \var\Big(\frac{n+2}{\epsilon_N^2} \rsc(x_N) \mid \deg(x_N) = k \Big)\,.
\end{align*}
By \Cref{eq:rsc_var_1} and \Cref{eq:rsc_var_2}, there is a constant $C$ such that if $k > 0$, then
\begin{equation}\label{eq:var}
    \var\Big( \frac{n+2}{\epsilon_N^2} \rsc(x_N) \mid \deg(x_N) = k \Big) \leq \frac{C}{k}\,.
\end{equation}
Therefore,
\begin{align*}
      &\p\Bigg[ \Big| \frac{n+2}{\epsilon_N^2} \cdot \rsc(x_N) - S(x_N) \Big| > \delta \Bigg]\\
      &\qquad\qquad\qquad \leq \sum_{k = 0}^{\infty} \p\Bigg[ \Big| \frac{n+2}{\epsilon_N^2} \cdot \rsc(x_N) - S(x_N) \Big| > \delta \mid \deg(x_N) = k \Bigg] \cdot \p[\deg(x_N) = k]\\
      &\qquad\qquad\qquad \leq \p[\deg(x_N) = 0] + \sum_{k=1}^{\infty} \frac{4C}{\delta^2k} \cdot \p[\deg(x_N) = k]\,.
\end{align*}
Choose $k_*$ such that $C\delta^{-2} k_*^{-1} < \xi/2$. Then
\begin{align*}
     \p\Bigg[ \Big| \frac{n+2}{\epsilon_N^2} \cdot \rsc(x_N) - S(x_N) \Big| > \delta \Bigg] &\leq \p[\deg(x_N) = 0] + \frac{C}{\delta^2} \cdot \p[k_N \leq k_*] + \sum_{k > k_*} \frac{C}{\delta^2k} \cdot \p[k_N = k] \\
     &\leq (1 + C\delta^{-2}) \cdot \p[k_N \leq k_*] + C\delta^{-2} k_*^{-1} \cdot \p[k_N > k] \\
     &\leq (1 + C\delta^{-2}) \cdot \p[k_N \leq k_*] + \xi/2\,.
\end{align*}
By \Cref{lem:degree_large}, $(1 + C\delta^{-2}) \cdot \p[k_N \leq k_*] < \xi/2$ for sufficiently large $N$. Consequently, $\frac{n+2}{\epsilon^2} \cdot \rsc(x)$ converges in probability to $S(x)$ as $N \to \infty$.
\end{proof}

\section{Convergence rate of Ollivier-Ricci graph curvature}\label{sec:modified_orc_convergence}

The primary aim of the paper is to prove that a node's scalar Ollivier-Ricci curvature (\cref{eq:orsc}) converges to scalar curvature in a suitable sense (\Cref{thm:main}). In \Cref{sec:mean_ricci_curvature}, we completed the first step by showing that a node's mean incident Ricci curvature $\rsc(x)$ converges to scalar curvature $S(x)$. To prove \Cref{thm:main}, we will show that $\orsc{x}$ converges to $\rsc(x)$, up to a multiplicative constant. To do this, we must show that Ollivier-Ricci graph curvature is a sufficiently good approximation to Ricci curvature. For a pair $x, y$ of adjacent nodes, \Cref{thm:KM_convergence} shows that the error in the approximation is
\begin{equation*}
    |\epsilon^{-2}\orc{x}{y} - \Ricc{x}{v_{xy}}| = \frac{|\W{G}(\mu_x^G, \mu_y^G) - \W{M}(\mu_x^M, \mu_y^M)|}{\epsilon^3} \cdot \frac{\epsilon}{d_M(x, y)} + \mathcal{O}(\epsilon + d_M(x, y))\,,
\end{equation*}
where $v_{xy} = \log_x(y) \in T_xM$. Our proof of \Cref{thm:main} requires us to 
\begin{enumerate}
    \item  prove that \[\max_{(x, y) \in E(G_N)}\frac{|\W{G}(\mu_x^G, \mu_y^G) - \W{M}(\mu_x^M, \mu_y^M)|}{\epsilon_N^3} \to 0\] with high probability and
    \item control the ratio $\epsilon/d_M(x, y)$, which may be arbitrarily large.
\end{enumerate}

In this section, we focus on (1), which we prove in \Cref{lem:WM_convergence_modified}. In addition to showing that this quantity  converges to $0$ with high probability, we additionally give a bound on the rate of convergence. As immediate consequences of this intermediate work, we prove two new results (\Cref{prop:orc_convergence1} and \Cref{prop:orc_convergence2}) related to the convergence of $\kappa_G(x, y)$ to the underlying manifold's Ricci curvature, although we do not directly use these results in our proof of \Cref{thm:main}.

Our proof of \Cref{lem:WM_convergence_modified} is broken up into two main parts. First, we control the error that is due to the discrepancy between the geodesic metric on $M$ and the shortest-path metric on $G_N$. We bound $|W_1^G(\mu_x^G, \mu_y^G) - W_1^M(\mu_x^G, \mu_y^G)|$ by applying results from \cite{bernstein_2000} to bound the graph-distance discrepancy $|d_G(u, v) - d_M(u, v)|$. Second, we bound $|W_1^M(\mu_x^G, \mu_y^G) - W_1^M(\mu_x^M, \mu_y^M)|$ by bounding $\max_{z \in G} W_1^M(\mu_z^G, \mu_z^M)$, the $1$-Wasserstein distance between the measure $\mu_z^G$ on neighboring nodes and the measure $\mu_z^M$ on the manifold ball $B^M(z, \epsilon_N)$.

\subsection{Bounding $|W_1^G(\mu_x^G, \mu_x^G) - W_1^M(\mu_x^G, \mu_y^G)|$}

We first recall two results from Bernstein et al \cite{bernstein_2000}, which together quantify the distortion in the graph-distance $d_G$ approximation to geodesic distance $d_M$. Let $V$ denote the volume of $M$.
\begin{lemma}\label{lem:bernstein_sampling}[Sampling Lemma in \cite{bernstein_2000}]
    Let $G$ be an RGG with connection threshold $\epsilon$, and let $\beta > 0$, $\lambda > 0$. Then $M \subseteq \bigcup_{z \in G} B^M(z, \lambda)$ with probability at least $1 - \beta$, provided that
    \begin{equation*}
        \frac{N}{\text{vol}(M)} > \log(\text{vol}(M)/\beta V_{\min}(\lambda/4))/V_{\min}(\lambda/2)\,,
    \end{equation*}
    where $V_{\min}(r) = \min_{x \in M} \text{vol}(B^M(x, r))$.
\end{lemma}
\noindent As noted in \cite{bernstein_2000}, the quantity $V_{\min}(r)$ is positive because $M$ is compact.
\begin{theorem}\label{thm:bernstein_geodesic_approx}[Theorem 2 in \cite{bernstein_2000}]
Let $G$ be an RGG with connection threshold $\epsilon$. Suppose that $4\lambda < \epsilon$. If $M \subseteq \bigcup_{z \in G} B^M(z, \lambda)$, then
\begin{equation*}
    d_M(u, v) \leq d_G(u, v) \leq (1 + 4\lambda/ \epsilon) d_M(u, v)
\end{equation*}
 for all pairs $u, v \in G$.
\end{theorem}
\noindent Bernstein et al \cite{bernstein_2000} apply \Cref{lem:bernstein_sampling} and \Cref{thm:bernstein_geodesic_approx} to obtain a probabilistic bound on the distortion of the graph-distance approximation to geodesic distance.

In \Cref{lem:modified_sampling} and \Cref{lem:modified_geodesic_approx} below, we apply the results of Bernstein et al \cite{bernstein_2000} to obtain a probabilistic bound on the graph-distance discrepancy as $N \to \infty$ and $\epsilon_N = N^{-\alpha} \to 0$.
\begin{lemma}\label{lem:modified_sampling}
     Let $\beta_N = N^{-\theta}$ and $\lambda_N = N^{-\omega}$, where $\theta, \omega > 0$ and $\theta + 2n \omega < 1$. Then for sufficiently large $N$, we have that $M \subseteq \bigcup_{z \in G_N} B^M(z, \lambda_N)$ with probability at least $1 - \beta_N$.
\end{lemma}
\begin{proof}
    In order to apply \Cref{lem:bernstein_sampling} (Sampling Lemma in \cite{bernstein_2000}), we must show that
    \begin{equation}\label{eq:sampling_bound1}
        \frac{\log(\text{vol}(M)/\beta_N V_{\min}(\lambda_N/4))}{N \cdot V_{\min}(\lambda_N/2)} < \fracc{\text{vol}(M)}
    \end{equation}
    for sufficiently large $N$. It suffices to show that the left-hand side is $o(1)$ as $N \to \infty$.

    For any $x \in M$, the volume of a geodesic ball $B^M(x, r)$ is characterized by the scalar curvature $S(x)$:
    \begin{equation*}
        \frac{\vol(B^M(x, r))}{v_n r^n} = 1 - \frac{6}{n+2}S(x)r^2 + \mathcal{O}(r^4)
    \end{equation*}
    as $r \to 0$, where we recall that $n = \dim(M)$. Because $M$ is compact, the scalar curvature is bounded, so $\vol(B^M(x, r)) \geq \fracc{2}v_n r^n$ for all $x \in M$ and sufficiently small $r$. Therefore there are positive constants $C_1$, $C_2$ such that
    \begin{equation}\label{eq:sampling_bound2}
         \frac{\log(\text{vol}(M)/\beta_N V_{\min}(\lambda_N/4))}{N \cdot V_{\min}(\lambda_N/2)} \leq C_1 \cdot\frac{\log(C_2/\beta_N \lambda_N^n)}{N \lambda_N^n}
    \end{equation}
    for sufficiently large $N$. We have
    \begin{equation*}
        \log(C_2/\beta_N \lambda_N^n) = o(1/\beta_N \lambda_N^n)
    \end{equation*}
    as $N \to \infty$ because $\log(x)/x \to 0$ as $x \to \infty$ and $\beta_N^{-1}\lambda_N^{-n} = N^{\theta + n\omega} \to \infty$ as $N \to \infty$. Therefore,
    \begin{equation*}
        C_1 \cdot\frac{\log(C_2/\beta_N \lambda_N^n)}{N \lambda_N^n} = o\Big( \fracc{N \beta_N\lambda_N^{2n}} \Big)\,.
    \end{equation*}
    By choice of $\theta$ and $\omega$, we have $N \beta_N \lambda_N^{2n} = N^{1 - \theta - 2n \omega} \to \infty$ as $N \to \infty$. Together with \Cref{eq:sampling_bound2}, this shows that \Cref{eq:sampling_bound1} holds for sufficiently large $N$. Applying \Cref{lem:bernstein_sampling} completes the proof.
\end{proof}

In \Cref{lem:modified_geodesic_approx}, we apply the previous results to obtain a bound on the graph-distance discrepancy for pairs $u, v$ of nodes that are neighbors of adjacent nodes $x, y$, respectively. These are the only relevant graph distances when we are calculating Ollivier-Ricci curvature. In the statement below, the parameter $\alpha < \fracc{6n}$ controls the connection threshold $\epsilon_N = N^{-\alpha}$, and the parameter $\beta_N = N^{-\theta}$, where $\theta \in (0, 1 - 6n\alpha)$, is the probability that the bound does not hold. The parameter $\lambda_N = N^{-\omega}$, where $\omega \in (3\alpha, \frac{1-\theta}{2n})$, controls the bound itself. The interval for $\theta$ is nonempty because $\alpha < \fracc{6n}$, and the interval for $\omega$ is nonempty because $\theta < 1 - 6n\alpha$. Increasing the probability $1 - \beta_N$ (equivalently, increasing $\theta$) comes at the cost of increasing the bound $\lambda_N$ (equivalently, decreasing $\omega$).
\begin{lemma}\label{lem:modified_geodesic_approx}
   Let $\beta_N = N^{-\theta}$ and $\lambda_N = N^{-\omega}$, where $\theta \in (0, 1 - 6n\alpha)$ and $\omega \in (3\alpha, \frac{1 - \theta}{2n})$. Then for sufficiently large $N$,
    \begin{equation*}
        \max_{(x, y) \in E(G_N)} \max_{\substack{u \in B^G(x, \epsilon_N) \\ v \in B^G(y, \epsilon_N)}} |d_M(u, v) - d_G(u, v)| \leq 12 \lambda_N \,,
    \end{equation*}
    with probability at least $1 - \beta_N$, where the right-hand side is $o(\epsilon_N^3)$ because $\omega > 3\alpha$.
\end{lemma}
\begin{proof}
For sufficiently large $N$ and probability at least $1 - \beta_N$, we have $M \subseteq \bigcup_{z \in G_N} B^M(z, \lambda_N)$ by \Cref{lem:modified_sampling}. Additionally, $4\lambda_N < \epsilon_N$ for sufficiently large $N$ because $\omega > \alpha$. By \Cref{thm:bernstein_geodesic_approx} (Theorem 2 in \cite{bernstein_2000}), we have
\begin{equation}\label{eq:geo_approx1}
    d_G(u, v) \leq (1 + 4\lambda_N/ \epsilon_N) d_M(u, v)
\end{equation}
for all $u, v \in G_N$. For any adjacent vertices $x, y \in G_N$ and any $u \in B^G(x, \epsilon_N)$, $v \in B^G(y, \epsilon_N)$, we have
\begin{equation*}
    d_M(u, v) \leq d_M(x, u) + d_M(x, y) + d_M(y, v) \leq 3\epsilon_N\,.
\end{equation*}
Therefore,
\begin{equation*}
    |d_M(u, v) - d_G(u, v)| \leq 12 \lambda_N
\end{equation*}
by \Cref{eq:geo_approx1}.
\end{proof}

\Cref{lem:G_to_M}, below, shows that $W_1^G(\mu_x^G, \mu_y^G)$ is a good approximation to $W_1^M(\mu_x^G, \mu_y^G)$, which concludes the first part of \Cref{sec:modified_orc_convergence}.  The work here is analogous to Appendix A.1 of \cite{ORC_convergence_hoorn}, with some minor corrections.

\begin{lemma}\label{lem:G_to_M}
Let $\beta_N = N^{-\theta}$, where $\theta \in (0, 1 - 6n\alpha)$, and let $\omega \in (3\alpha, \frac{1 - \theta}{2n})$. Then for sufficiently large $N$,
    \begin{equation*}
    \max_{(x, y) \in E(G_N)}\frac{|\W{M}(\mu_x^G, \mu_y^G) - \W{G}(\mu_x^G, \mu_{y}^G)|}{\epsilon_N^3} \leq \frac{12}{N^{\omega - 3\alpha}} = o(1)
\end{equation*}
with probability at least $1 - \beta_N$.
\end{lemma}
\begin{proof}
    Let $\lambda_N = N^{-\omega}$. By \Cref{lem:modified_geodesic_approx},
    \begin{equation}\label{eq:distortion_bound}
        \max_{(x, y) \in E(G_N)}\max_{\substack{u \in B^G(x, \epsilon_N) \\ v \in B^G(y, \epsilon_N)}} |d_M(u, v) - d_G(u, v)| \leq 12\lambda_N\,,
    \end{equation}
    with probability at least $1 - \beta_N$. Therefore,
    \begin{equation*}
       \max_{(x, y) \in E(G_N)} \frac{|\W{M}(\mu_{x}^G, \mu_y^G) - \W{G}(\mu_{x}^G, \mu_y^G)|}{\epsilon_N^3} \leq \frac{12 \lambda_N }{\epsilon_N^3} = \frac{12}{N^{\omega - 3\alpha}}\,.
    \end{equation*}
\end{proof}

\subsection{Bounding $|W_1^M(\mu_x^G, \mu_x^G) - W_1^M(\mu_x^M, \mu_y^M)|$}
As outlined at the beginning, the second part of this section is to bound the error that is introduced by the Wasserstein distance between $\mu_z^G$ and $\mu_z^M$, where $z$ is any node in $G$. We start by recalling the following result from \cite{trillos_2020}, where we define $p_n := \begin{cases} 3/4\,, & n = 2\\ 1/n\,, & n \geq 3\end{cases}$.
\begin{theorem}[Theorem 2 in \cite{trillos_2020}]\label{thm:trillos}
Let $\mu$ be the uniform probability measure on $M$ and $\mu_N$ the empirical probability measure $\mu_N$ on the nodes of $G_N$. Then for any $s > 0$, there are constants $A_{M, s}'$, $C_{M, s}$ and a transport map $T_N : M \to X_N$ such that
\begin{equation*}
    \sup_{x \in M} d(x, T_N(x)) \leq A_{M, s}' \frac{\log(N)^{p_n}}{N^{1/n}}
\end{equation*}
with probability at least $1 - C_{M, s}N^{-s}$
\end{theorem}
The benefit of \Cref{thm:trillos} is that we can leverage the transport map $T_N$ to obtain a bound on $W_1^M(\mu_z^M, \mu_z^G)$. However, note that $T_N$ does not necessarily map $B^M(z, \epsilon_N)$ to $B^G(z, \epsilon_N)$, so it does not immediately provide a transport map for $\mu_z^M, \mu_z^G$.
\begin{lemma}\label{lem:continuous_to_discrete} For any $s > 0$, there is a constant $A_{M, s}$ such that
    \begin{equation*}
        \max_{z \in G_N} \W{M}(\mu_z^G, \mu_z^M) \leq A_{M, s} \frac{\log(N)^{p_n}}{N^{1/n}}
    \end{equation*}
    with probability at least $1 - C_{M, s}N^{-s}$, where $C_{M, s}$ is the constant from \Cref{thm:trillos}.
\end{lemma}
\begin{proof}
Let $\mu$ be the uniform probability measure on $M$ and let $\mu_N$ be the uniform probability measure on the nodes $X_N$ of $G_N$. Suppose that $s > 0$. By \Cref{thm:trillos}, there are constants $A_{M, s}', C_{M, s}$ and a transport map $T_N: M \to X_N$ such that
\begin{equation*}
    D_N := \sup_{x \in M} d(x, T_N(x)) \leq A_{M, s}' \frac{\log(N)^{p_n}}{N^{1/n}}
\end{equation*}
with probability at least $1 - C_{M, s} N^{-s}$.

    Let $z$ be any node in $G_N$, and define
    \begin{align*}
        B_N&:= \BM{z}{\epsilon_N}\,, \\
        \hat{B}_N &:= T_N^{-1}\Big(X_N \cap \BM{z}{\epsilon_N}\Big)\,.
    \end{align*}
    Let $\hat{\mu}_z^M$ be the uniform measure on $\hat{B}_N$. We have
    \begin{equation*}
        W^M_1(\mu_z^G, \mu_z^M) \leq W^M_1(\mu_z^G, \hat{\mu}_z^M) + W^M_1(\hat{\mu}_z^M, \mu_z^M) \leq W^M_{\infty}(\mu_z^G, \hat{\mu}_z^M) + W^M_1(\hat{\mu}_z^M, \mu_z^M)\,.
    \end{equation*}
    We begin by upper bounding $\Winf(\mu_z^G, \hat{\mu}_z^M)$. Restricting $T_N$ to $\hat{B}_N$ yields a transport map from $\hat{B}_N$ to $X_N \cap \BM{z}{\epsilon_N}$ such that 
    \begin{equation*}
        \sup_{x\in \hat{B}_N} d_M(T_N(x), x) \leq D_N \leq A_{M, s}' \frac{\log(N)^{p_n}}{N^{1/n}}\,,
    \end{equation*}
    which implies that
    \begin{equation*}
        W^M_{\infty}(\mu_z^G, \hat{\mu}_z^M) \leq A'_{M, s} \frac{\log(N)^{p_n}}{N^{1/n}}\,.
    \end{equation*}
    Next, we upper bound $W_1^M(\hat{\mu}_z^M, \mu_z^M)$. Define
    \begin{align*}
        \epsilon_N^{\pm} &:= \epsilon_N \pm D_N\,, \\
        B^{\pm}_N &:= \BM{z}{\epsilon_N^{\pm}}\,.
    \end{align*}
    We must have that 
    \begin{equation}\label{eq:ball_bounds}
        B^-_N \subseteq \hat{B}_N \subseteq B^+_N\,.
    \end{equation}
    However, we note that the masses 
    \begin{align*}
        \hat{\mu}_z^M(B^-_N) &= \frac{\mu(B^-_N)}{\mu(\hat{B}_N)}\,, \\
        \mu_z^M(B^-_N) &= \frac{\mu(B^-_N)}{\mu(B_N)}
    \end{align*}
    generally differ because generally $\mu(\hat{B}_N) \neq \mu(B_N)$.
    
    Let $\gamma \in \Gamma(\hat{\mu}_z^M, \mu_z^M)$ be a coupling that fixes mass $m_N := \min\Big(\hat{\mu}_z^M(B^-_N), \mu_z^M(B^-_N)\Big)$ in set $B^-_N$. Therefore,
    \begin{equation*}
        W_1^M(\hat{\mu}^M_z, \mu_z^M) \leq \int_{\hat{B}_N} \int_{B_N} d_M(x, y) \,d\gamma(x, y) 
        \leq \text{diam}(B_N^+) \cdot (1 - m_N) = 2 \epsilon_N^+ \cdot (1 - m_N)
    \end{equation*}
    because $d_M(x, y) \leq \text{diam}(B_N^+)$ for all $x$ and $y$, and $1 - m_N$ is the total amount of mass that $\gamma$ moves a nonzero distance. We note that 
    \begin{equation*}
        \epsilon_N^+ - \epsilon_N = D_N = \mathcal{O}\Big( \frac{\log(N)^{p_n}}{N^{1/n}}\Big) = o(N^{-\alpha}) = o(\epsilon_N)
    \end{equation*}
    because $\alpha < 1/n$, so we can replace $\epsilon_N^+$ with $\mathcal{O}(\epsilon_N)$ to obtain
    \begin{equation}\label{eq:W1_bound_setup}
        W_1(\hat{\mu}_z, \mu_z^M) = \mathcal{O}(\epsilon_N) \cdot (1 - m_N)\,.
    \end{equation}
    The rest of the proof is devoted to bounding $(1 - m_N)$. By \Cref{eq:ball_bounds},
    \begin{equation}\label{eq:hatB_vol_bounds}
         \mu(B_N) - \mu(B_N\setminus B^-_N) \leq \mu(\hat{B}_N) \leq \mu(B_N) + \mu(B^+_N \setminus B_N)\,.
    \end{equation}
    Using the fact that $|s^{-1} - r^{-1}| \leq \min(s, r)^{-2} |s-r|$ for all $s, r > 0$, we bound the difference in masses $\hat{\mu}_z^M(B_N^-)$, $\mu_z^M(B_N^-)$ by
    \begin{align}
        |\hat{\mu}_z^M(B^-_N) -  \mu_z^M(B^-_N)| &\leq \mu(B^-_N)|\mu(\hat{B}_N)^{-1} - \mu(B_N)^{-1}| \notag \\
        &\leq \mu(B^-_N)\min(\mu(\hat{B}_N), \mu(B_N))^{-2} |\mu(\hat{B}_N) - \mu(B_N)| \notag \\
        &\leq \mu(B^-_N)^{-1} \max(\mu(B^+_N \setminus B_N), \mu(B_N\setminus B^-_N)) \label{eq:mass_difference}\,,
    \end{align}
    where the last inequality follows from \Cref{eq:hatB_vol_bounds}. The sets $B^+_N \setminus B_N$ and $B_N \setminus B^-_N$ are annuli in $M$ such that the difference between the inner and outer radii is $D_N$ for both annuli. Therefore, as $\epsilon_N \to 0$,
    \begin{align}
        \mu(B^+_N\setminus B_N) &= D_N \cdot \mathcal{O}((\epsilon_N^+)^{n-1})\,,\\
        \mu(B_N \setminus B^-_N) &= D_N \cdot \mathcal{O}(\epsilon_N^{n-1}) \label{eq:inner_annulus_volume}
    \end{align}
    We recall again that $\epsilon_N^+ - \epsilon_N = o(\epsilon_N)$, so we can replace $\mathcal{O}((\epsilon_N^+)^{n-1})$ by $\mathcal{O}(\epsilon_N^{n-1})$. Therefore by \Cref{eq:mass_difference},
    \begin{equation*}
        |\hat{\mu}_z^M(B^-_N) -  \mu_z^M(B^-_N)| \leq \frac{D_N}{\mu(B^-_N)} \cdot \mathcal{O}(\epsilon_N^{n-1})\,.
    \end{equation*}
    As $\epsilon_N \to 0$, we have $\mu(B_N^-)^{-1} = \mathcal{O}((\epsilon_N^-)^{-n})$, which is $\mathcal{O}(\epsilon_N^{-n})$ because $\epsilon_N - \epsilon_N^- = o(\epsilon_N)$. Therefore,
    \begin{equation*}
        |\hat{\mu}_z^M(B^-_N) -  \mu_z^M(B^-_N)| \leq D_N \cdot \mathcal{O}(\epsilon_N^{-1})\,.
    \end{equation*}
    Consequently,
    \begin{align}
        1 - m_N  &\leq 1- \mu_z^M(B_N^-) + |\hat{\mu}_z^M(B_N^-) - \mu_z^M(B_N^-)| \notag \\
        &\leq 1 - \mu_z^M(B^-_N) + D_N \cdot \mathcal{O}(\epsilon_N^{-1})\notag \\
        &= \mu_z^M(B_N \setminus B^-_N) + D_N \cdot \mathcal{O}(\epsilon_N^{-1}) \label{eq:mn_bound}\,,
    \end{align}
    where we recall that $m_N = \min\Big(\hat{\mu}_z^M(B^-_N), \mu_z^M(B^-_N)\Big)$. By \Cref{eq:inner_annulus_volume},
    \begin{equation*}
        \mu_z^M(B_N \setminus B^-_N) = \frac{\mu(B_N \setminus B_N^-)}{\mu(B_N^-)} = D_N \cdot \mathcal{O}(\epsilon_N^{-1})\,,
    \end{equation*}
    where we use the fact that $\mu(B_N^-)^{-1} = \mathcal{O}(\epsilon_N^{-n})$. Therefore, by \Cref{eq:mn_bound}
    \begin{equation*}
        1 - m_N \leq D_N \cdot \mathcal{O}(\epsilon_N^{-1})\,,
    \end{equation*}
    so by \Cref{eq:W1_bound_setup}, there is a constant $C$ that does not depend on $z$ such that
    \begin{equation*}
        W_1(\hat{\mu}_z^M, \mu_z^M) \leq C \cdot D_N \leq C \cdot A'_{M, s} \frac{\log(N)^{p_n}}{N^{1/n}}\,.
    \end{equation*}
\end{proof}

We are now ready to bound $\max_{(x, y) \in E(G_N)} |W_1^M(\mu_x^M, \mu_y^M) - W_1^G(\mu_x^G, \mu_y^G)|$, which is the main purpose of this section. The two crucial ingredients are \Cref{lem:G_to_M} and \Cref{lem:continuous_to_discrete}.

\begin{lemma}\label{lem:WM_convergence_modified} Let $\theta \in (0, 1 - 6n\alpha)$. Then there are constants $A_{M, \theta}, C_{M, \theta}$ such that for sufficiently large $N$,
\begin{equation*}
    \max_{(x, y) \in E(G_N)} \frac{|\W{M}(\mu_x^M, \mu_y^M) - \W{G}(\mu_x^G, \mu_y^G)|}{\epsilon_N^3} \leq 3A_{M, \theta} \frac{\log(N)^{p_n}}{N^{1/n - 3\alpha}}
\end{equation*}
with probability at least $1 - (1 + C_{M, \theta})N^{-\theta}$.
\end{lemma}
\begin{proof}
We have
\begin{align*}
    &\max_{(x, y) \in E(G_N)} \frac{|\W{M}(\mu_{x}^M, \mu_y^M) - \W{G}(\mu_{x}^G, \mu_y^G)|}{\epsilon_N^3} \\
&\qquad\qquad \leq \max_{(x, y) \in E(G_N)}\frac{|\W{M}(\mu_{x}^G, \mu_y^G) - \W{G}(\mu_{x}^G, \mu_y^G)|}{\epsilon_N^3} + \frac{|W_1^M(\mu_x^G, \mu_y^G) - W_1^M(\mu_x^M, \mu_y^M)|}{\epsilon_N^3} \\
&\qquad\qquad\leq \max_{(x, y) \in E(G_N)}\frac{|\W{M}(\mu_{x}^G, \mu_y^G) - \W{G}(\mu_{x}^G, \mu_y^G)|}{\epsilon_N^3} + \frac{\W{M}(\mu_x^G, \mu_x^M) + \W{M}(\mu_y^G, \mu_y^M)}{\epsilon_N^3} \\
&\qquad\qquad\leq \max_{(x, y) \in E(G_N)}\frac{|\W{M}(\mu_{x}^G, \mu_y^G) - \W{G}(\mu_{x}^G, \mu_y^G)|}{\epsilon_N^3} + 2\max_{z \in G_N} W_1(\mu_z^G, \mu_z^M)\,.
\end{align*}
Choose $\omega \in (3\alpha, \frac{1 - \theta}{2n})$, which is a nonempty interval because $\theta < 1 - 6n\alpha$.
By \Cref{lem:G_to_M}, 
\begin{equation*}
    \max_{(x, y) \in E(G_N)}\frac{|\W{M}(\mu_{x}^G, \mu_y^G) - \W{G}(\mu_{x}^G, \mu_y^G)|}{\epsilon_N^3} \leq \frac{12}{N^{\omega - 3\alpha}}
\end{equation*}
with probability at least $1 - N^{-\theta}$. By \Cref{lem:continuous_to_discrete}, there are constants $A_{M, \theta}, C_{M, \theta}$ such that
\begin{equation*}
    \max_{z \in G_N} W_1(\mu_z^G, \mu_z^M) \leq A_{M, \theta} \frac{\log(N)^{p_n}}{N^{1/n - 3\alpha}}
\end{equation*}
with probability at least $1 - C_{M, \theta}N^{-\theta}$. Therefore,
\begin{equation*}
    \max_{(x, y) \in E(G_N)} \frac{|\W{M}(\mu_{x}^M, \mu_y^M) - \W{G}(\mu_{x}^G, \mu_y^G)|}{\epsilon_N^3} \leq \frac{12}{N^{\omega - 3\alpha}} + 2A_{M, \theta} \frac{\log(N)^{p_n}}{N^{1/n - 3\alpha}}
\end{equation*}
with probability at least $1 - (1 + C_{M, \theta})N^{-\theta}$. For sufficiently large $N$, we have $\frac{12}{N^{\omega - 3\alpha}} < A_{M, \theta} \frac{\log(N)^{p_n}}{N^{1/n - 3\alpha}}$ because $\omega < 1/n$.
\end{proof}

\subsection{Bounds on the error between Ollivier-Ricci curvature and Ricci curvature}
\Cref{lem:WM_convergence_modified} is all that we need for \Cref{thm:main}, the main result of our paper. However, we can also use it to derive two new convergence results about Ollivier-Ricci curvature. 

The first, \Cref{prop:orc_convergence1}, proves that we can bound the rate of convergence $\kappa_G(x, y) \to \Ricc{x}{v_{xy}}$, in terms of $d_M(x, y)$, for all edges. For $d_M(x, y)$ close to $\epsilon_N$, the largest possible value, the bound is strongest. As $d_M(x, y) \to 0$, the bound becomes weaker; in other words, $\kappa_G(x, y)$ has little relation to Ricci curvature when $(x, y)$ is very short relative to the connection threshold $\epsilon_N$. Our intuition is that when $d_M(x, y)$ is very small, the $1$-hop neighborhoods of $x$ and $y$ will have substantial overlap, leading to positive Ollivier-Ricci curvature regardless of the manifold's curvature.
\begin{proposition}\label{prop:orc_convergence1}
    Let $\theta \in (0, 1 - 6n\alpha)$. There are constants $C_{1, \theta}, C_{2, \theta}$ such that with probability at least $1 - C_{1, \theta}N^{-\theta}$,
    \begin{equation*}
        |\epsilon_N^{-2} \orc{x}{y} - \Ricc{x}{v_{xy}}| \leq C_{2, \theta} \Big( \epsilon_N +  \frac{\log(N)^{p_n}}{d_M(x, y) N^{1/n - 2\alpha}}\Big)
    \end{equation*}
    for all edges $(x, y) \in E(G_N)$, where $v_{xy} = \log_x(y)/\norm{\log_x(y)}$.
\end{proposition}
\begin{proof}
By \Cref{thm:KM_convergence},
\begin{equation*}
    |\epsilon_N^{-2}\kappa_M(x, y) - \Ricc{x}{v_{xy}}| = \mathcal{O}(\epsilon_N)\,.
\end{equation*}
By \Cref{lem:WM_convergence_modified}, there are constants $A_{M, \theta}$, $C_{M, \theta}$ such that
\begin{align*}
    \max_{(x, y) \in E(G_N)}\epsilon_N^{-2}|\orc{x}{y} - \kappa_M(x, y)| &=  \max_{(x, y) \in E(G_N)} \frac{|\W{M}(\mu_x^M, \mu_y^M) - \W{G}(\mu_x^G, \mu_y^G)|}{\epsilon_N^3} \cdot \frac{\epsilon_N}{d_M(x, y)} \\
    &\leq 3A_{M, \theta}\frac{\log(N)^{p_n}}{d_M(x, y)N^{1/n - 2\alpha}}
\end{align*}
with probability at least $1 - (1 + C_{M, \theta})N^{-\theta}$.
\end{proof}
\Cref{prop:orc_convergence1} showed that a bound holds for \textit{all} edges with high probability as $N \to \infty$. For comparison, our second result (\Cref{prop:orc_convergence2}) says that the Ollivier-Ricci curvature of a \emph{random}\footnote{The edge is not chosen uniformly at random. We choose one endpoint uniformly at random, then one of its neighbors uniformly at random.} edge converges in probability to the Ricci curvature. It requires \Cref{lem:ratio_convergence_dist}, below, which will also be used later in our proof of \Cref{thm:main}. The lemma says that the ratio $\frac{\epsilon_N}{d_M(x, y)}$ converges in distribution to the analogous random variable in Euclidean space as $\epsilon_N \to 0$.
\begin{lemma}\label{lem:ratio_convergence_dist}
    Suppose that $x_N$ is chosen uniformly at random from $M$ and $y_N$ is chosen uniformly at random from $\BM{x_N}{\epsilon_N}$. Let $Z_N = \frac{\epsilon_N}{d_M(x_N, y_N)}$, and let $Z$ be the random variable with cdf 
    \begin{equation}\label{eq:Z_cdf}
        F_Z(z) = \begin{cases}
            0\,, & z < 1\,,\\
            1 - \fracc{z^n}\,, & z \geq 1\,.
        \end{cases}
    \end{equation}
    Then $Z_N \to Z$ in distribution. Moreover, there is a constant $C > 0$ such that
    \begin{equation*}
        |\p[Z_N \leq z] - F_Z(z)| \leq C \cdot \frac{\epsilon_N^2}{z^n}
    \end{equation*}
    for all $z$ and all sufficiently large $N$.
\end{lemma}
\begin{proof}
    When $z < 1$, we have $\p[Z_N \leq z] = 0$ because $\epsilon_N \geq d_M(x_N, y_N)$. When $z \geq 1$,
    \begin{equation*}
        \p[ Z_N \leq z] = 1 - \p[d_M(x_N, y_N) \leq \epsilon_N/z] = 1 - \frac{\vol(\BM{x_N}{\epsilon_N/z})}{\vol(\BM{x_N}{\epsilon_N})}\,.
    \end{equation*}
    Because $M$ is compact, there is a constant $C'$ such that
    \begin{equation}\label{eq:volume_bound}
        |\vol(\BM{x}{r}) - v_nr^n| \leq C'r^{n+2}
    \end{equation}
    for sufficiently small $r$ and all $x \in M$. Therefore, because $\epsilon_N/z \leq \epsilon_N$ for $z \geq 1$ and $\epsilon_N \to 0$ as $N \to \infty$, we have
    \begin{align*}
        |\p[Z_N \leq z] - F_Z(z)| &= \Big\vert \frac{\vol(\BM{x_N}{\epsilon_N/z})}{\vol(\BM{x_N}{\epsilon_N})} - \fracc{z^n} \Big\vert \\
        &\leq \Big\vert \frac{v_n(\epsilon_N/z)^n + C'(\epsilon_N/z)^{n+2}}{v_n \epsilon_N^n + C'\epsilon_N^{n+2}} - \fracc{z^n} \Big\vert \\
        &= \Big\vert \frac{z^{-n} + \mathcal{O}(\epsilon_N^2)/z^{n+2}}{1 + \mathcal{O}(\epsilon_N^2)} - z^{-n} \Big \vert\\
        &\leq \mathcal{O}(\epsilon_N^2) \cdot \Big(z^{-n} + \mathcal{O}(\epsilon_N^2)/z^{n+2}\Big) + \mathcal{O}(\epsilon_N^2)/z^{n+2} \\
        &\leq \mathcal{O}(\epsilon_N^2)z^{-n}
    \end{align*}
    as $N \to \infty$, for all $z \geq 1$. Consequently, $Z_N \to Z$ in distribution.
\end{proof}

\begin{proposition}\label{prop:orc_convergence2}
    Let $x_N$ be a node chosen uniformly at random from $G_N$ and let $y_N$ be an adjacent node chosen uniformly at random. As $N \to \infty$,
    \begin{equation*}
        \epsilon_N^{-2} \kappa_G(x_N, y_N) \xrightarrow[]{\p} \Ricc{x_N}{v_N} \,,
    \end{equation*}
    where $v_N = \log_{x_N}(y_N)/\norm{\log_{x_N}(y_N)}$.
\end{proposition}
\begin{proof}
    First, we note that 
    \begin{equation*}
        \epsilon_N^{-2}|\kappa_G(x, y) - \kappa_M(x, y)| = \Big\vert \frac{W_1^M(\mu_x^M, \mu_y^M) - W_1^G(\mu_x^G, \mu_y^G)}{\epsilon_N^3} \cdot \frac{\epsilon_N}{d_M(x, y)}\Big\vert\,.
    \end{equation*}
    By \Cref{lem:ratio_convergence_dist}, the ratio $\epsilon_N/d_M(x_N, y_N) \to 0$ in distribution as $N \to \infty$. By \Cref{lem:WM_convergence_modified},
    \begin{equation*}
        \frac{|\W{M}(\mu_{x_N}^M, \mu_{y_N}^M) - W_1^G(\mu_{x_N}^G, \mu_{y_N}^G)|}{\epsilon_N^3} \xrightarrow[]{\p} 0
    \end{equation*}
    as $N \to \infty$. Therefore,
    \begin{align*}
        |\epsilon_N^{-2} \orc{x_N}{y_N} - \epsilon_N^{-2}\kappa_M(x_N, y_N)| \xrightarrow[]{\p} 0
    \end{align*}
    by Slutsky's theorem. By \Cref{thm:KM_convergence}, $|\epsilon_N^{-2}\kappa_M(x_N, y_N) - \Ricc{x_N}{v_N}| \to 0$ as $N \to \infty$.
\end{proof}

\section{Main Results}\label{sec:mean_orc_convergence}
Finally, we prove that $\orsc{x}$ converges in a suitable sense to the scalar curvature $S(x)$. As we sketched at the beginning of \Cref{sec:modified_orc_convergence}, the last remaining ingredient is to control the size of the ratio $\epsilon/d_M(x, y)$ for pairs $x, y$ of adjacent nodes. The ratio is a random variable that takes values in $[1, \infty)$. We prove in \Cref{lem:ratio_sample_mean_all_n} that for a random node $x_N \in G_N$, the mean ratio $\fracc{deg(x_N)}\sum_{y \sim x_N} \frac{\epsilon_N}{d_M(x_N, y)}$ converges to a constant as $N \to \infty$. For dimensions $n \geq 3$, we do this by applying a concentration-inequality argument. Unfortunately, in dimension $n = 2$, the variance of the ratio is infinite (see \Cref{lem:infinite_variance}), so the proof requires more care.

\subsection{The ratio $\epsilon / d_M(x, y)$}\label{sec:ratio}

Suppose that $x_N$ is a node chosen uniformly at random from the graph $G_N$. We define
\begin{equation*}
    \ratiosumA := \sum_{i=1}^{\deg(x_N)} \frac{\epsilon_N}{d_M(x_N, y_N^i)}\,,
\end{equation*}
where $y_N^1, \ldots, y_N^{\deg(x_N)}$ are the neighbors of $x_N$.

The degree of $x_N$ is a random variable, so it is easier to study the sum of a fixed number $k$ of random variables with the same probability distribution. Because the neighbors $y$ of $x_N$ are uniformly distributed in $B^M(x_N, \epsilon_N)$, we have that for all $z \in \R$,
\begin{equation*}
    \p\Big[ \frac{\epsilon_N}{d_M(x_N, y)} \leq z \Big] = \p\Big[ \frac{\epsilon_N}{d_M(x_N, p)} \leq z\Big]\,,
\end{equation*}
where $p$ is a point sampled uniformly at random from $B^M(x_N, \epsilon_N)$. Letting $\{p_N^i\}_{i=1}^{\infty}$ be a sequence of points sampled uniformly at random from $B^M(x_N, \epsilon_N)$, we define
\begin{align*}
    Z_N^i &:= \frac{\epsilon_N}{d_M(x_N, p_N^i)}\\
    S_{N, k} &:= \sum_{i=1}^k Z_N^i\,.
\end{align*}

We define $\mu_N = \E[Z_N^1]$ to be the mean of one of the ratios. Recall that we defined the random variable $Z$ with CDF defined by \Cref{eq:Z_cdf}. Let $\mu = \E[Z]$.
\begin{lemma}\label{lem:expectation_convergence}
\begin{equation*}
    \lim_{N \to \infty} \mu_N = \mu = \frac{n}{n-1}\,.
\end{equation*}
\end{lemma}
\begin{proof}
    Because $Z_N^1$ and $Z$ are non-negative,
    \begin{equation*}
        \vert \mu_N- \mu \vert = \Big\vert \int_0^{\infty} F_Z(z) - \p[Z_N^1 \leq z] \,dz \Big\vert\,,
    \end{equation*}
    where $F_Z(z)$ is the cdf for $Z$ defined by \Cref{eq:Z_cdf}. We have $Z_N^1 \geq 1$, so $\p[Z_N^1 \leq z] = 0 = F_Z(z)$ for $z <1$. Therefore,
    \begin{equation*}
         \vert \mu_N - \mu \vert = \Big\vert \int_1^\infty F_Z(z) - \p[Z_N^1 \leq z]dz \Big\vert \leq \int_1^\infty \vert F_Z(z) - \p[Z_N^1 \leq z] \vert \,dz\,.
    \end{equation*}
    By \Cref{lem:ratio_convergence_dist}, there is a constant $C > 0$ such that
    \begin{equation*}
         \vert \mu_N - \mu \vert \leq C\epsilon_N^2 \int_1^\infty z^{-n}\,dz = \frac{C\epsilon_N^2}{n-1}
    \end{equation*}
    for sufficiently large $N$. The right-hand side approaches $0$ as $N \to \infty$ because $\epsilon_N \to 0$.

    The mean $\mu = \E[Z]$ is
    \begin{equation*}
        \E[Z] = \int_0^\infty (1 - F_Z(z))\,dz = 1 + \int_1^\infty z^{-n}\,dz = \frac{n}{n-1}\,.
    \end{equation*}
\end{proof}

In \Cref{lem:var_bounded}--\Cref{lem:ratio_sample_mean_n3}, we prove that $\ratiosumA/\deg(x_N) \xrightarrow[]{\p} \mu$ if the dimension $n \geq 3$. The proof for the special case $n = 2$ is postponed until afterwards.
\begin{lemma}\label{lem:var_bounded}
    If $n \geq 3$, then the sequence $\{\var(Z_N^1)\}_{N \in \mathbb{N}}$ is bounded.
\end{lemma}
\begin{proof}
    By \Cref{lem:expectation_convergence}, $\lim_{N \to \infty} \mu_N = \mu < \infty$, so it suffices to show that $\{\E[(Z_N^1)^2]\}_{N=1}^{\infty}$ is bounded. Because $(Z_N^1)^2$ and $Z^2$ are non-negative random variables,
    \begin{equation*}
        \vert \E[(Z_N^1)^2] - \E[Z^2] \vert = \Big\vert \int_1^\infty \Big(F_Z(\sqrt{z}) - \p[Z_N^1 \leq \sqrt{z}]\Big)\,dz \Big\vert\,.
    \end{equation*}
    By \Cref{lem:ratio_convergence_dist}, there is a constant $C$ such that for sufficiently large $N$,
    \begin{equation*}
         \vert \E[(Z_N^1)^2] - \E[Z^2] \vert \leq C\epsilon_N^2 \int_1^\infty z^{-n/2} \,dz = \frac{2\epsilon_N^2}{n-2}\,,
    \end{equation*}
    which converges to $0$ as $N \to \infty$ because $\epsilon_N \to 0$. Therefore, it suffices to show that $\E[Z^2] < \infty$. We can calculate
    \begin{equation*}
        \E[Z^2] = \int_0^\infty \p[Z^2 > z] \,dz = 1 + \int_1^\infty z^{-n/2} \,dz = \frac{n}{n-2} < \infty\,.
    \end{equation*}
\end{proof}

\begin{lemma}\label{lem:LLN_uniform}
    If $n \geq 3$, then the convergence $\Snk/k \xrightarrow[]{\p} \mu_N$ as $k \to \infty$ is uniform in $N$. That is, if $\eta, \xi > 0$, then
    \begin{equation*}
        \p[ |\Snk/k - \mu_N | > \eta] < \xi
    \end{equation*}
    for all $N$ and sufficiently large $k$.
\end{lemma}
\begin{proof}
    By Chebyschev's inequality,
    \begin{equation*}
        \p[ |\Snk/k - \mu_N | > \eta] \leq \frac{\var(\Snk/k)}{\eta^2/4} = \fracc{k} \cdot \frac{\var(Z_N^1)}{\eta^2/4}\,.
    \end{equation*}
    By \Cref{lem:var_bounded}, there is a constant $C$ such that $ \p[ |\Snk/k - \mu_N | > \eta] \leq C/k$ for all $N$. Therefore, $\Snk/k \xrightarrow[]{\p} \mu_N$ as $k \to \infty$, uniformly in $N$.
\end{proof}

\begin{lemma}\label{lem:ratio_sample_mean_n3}
    When $n \geq 3$, we have $\ratiosumA/\deg(x_N) \xrightarrow[]{\p} \mu = \frac{n}{n-1}$ as $N \to \infty$.
\end{lemma}
\begin{proof}
    Let $\eta > 0$ and $\xi > 0$. We wish to show that
    \begin{equation*}
        \lim_{N \to \infty} \p[ |\ratiosumA/\deg(x_N) - \mu| > \eta]  < \xi\,.
    \end{equation*}
    for sufficiently large $N$. We may expand $\p[ |\ratiosumA/\deg(x_N)  - \mu| > \eta$ as the sum
    \begin{align*}
        \p[ |\ratiosumA/\deg(x_N) - \mu_N | > \eta] &= \sum_{k=0}^{\infty} \p[|\ratiosumA/\deg(x_N) - \mu_N| > \eta \, \vert \deg(x_N) = k] \cdot \p[\deg(x_N) = k] \\
        &= \sum_{k = 0}^{\infty} \p[ |\Snk/k - \mu| > \eta] \cdot \p[\deg(x_N) = k]\,.
    \end{align*}
    For sufficiently large $N$, we have $|\mu_N - \mu| < \eta/2$ by \Cref{lem:expectation_convergence}, so
    \begin{equation*}
        \p[ |\Snk/k - \mu_N | > \eta] \leq \sum_{k = 0}^{\infty} \p[ |\Snk/k - \mu_N| > \eta/2] \cdot \p[\deg(x_N) = k]
    \end{equation*}
    By \Cref{lem:LLN_uniform}, there is a $k_*$ such that if $k > k_*$, then
    \begin{equation*}
        \p[ |\Snk/k - \mu_N| > \eta/2] < \xi
    \end{equation*}
    for all $N$. Therefore,
    \begin{align*}
        \p[ |\Snk/k - \mu_N | > \eta] &\leq \sum_{k \leq k_*} \p[ |\Snk/k - \mu_N| > \eta/2] \cdot \p[\deg(x_N) = k] \\
        &\hspace{20mm} + \sum_{k > k_*} \p[ |\Snk/k - \mu_N| > \eta/2] \cdot \p[\deg(x_N) = k] \\
        &\leq \p[\deg(x_N) \leq k_*] + \xi
    \end{align*}
    for sufficiently large $N$. \Cref{lem:degree_large} shows that $\lim_{N \to \infty} \p[\deg(x_N) \leq k_*] = 0$, which concludes the proof.
\end{proof}

Unfortunately, the proof of \Cref{lem:ratio_sample_mean_n3} does not generalize to $n = 2$ because $Z_N^1$ and $Z$ have infinite variance when $n = 2$, so we cannot apply Chebyschev's inequality as we did in the proof of \Cref{lem:LLN_uniform}.
\begin{lemma}\label{lem:infinite_variance}
    The random variables $Z$ and $Z_N^1$ have infinite variance for sufficiently large $N$.
\end{lemma}
\begin{proof}
    By \Cref{lem:expectation_convergence}, it suffices to show that $\E[Z^2]$ and $\E[(Z_N^1)^2]$ are infinite for sufficiently large $N$. For $Z$, this is true because
    \begin{equation*}
        E[Z^2] = \int_0^\infty \p[Z > \sqrt{z}] \,dz = 1 + \int_1^\infty z^{-1} \,dz\,,
    \end{equation*}
    and the integral diverges. For $Z_N^1$, we recall from the proof of \Cref{lem:ratio_convergence_dist} that 
    \begin{equation*}
        \p[Z_N^1 > \sqrt{z}] = \frac{\vol(\BM{x_N}{\epsilon_N/\sqrt{z}})}{\vol(\BM{x_N}{\epsilon_N})}
    \end{equation*}
    The ball-volume bound in \Cref{eq:volume_bound} implies that
    \begin{equation*}
        \E[(Z_N^1)^2] = 1 + \int_1^\infty \p[Z_N^1 > \sqrt{z}] \,dx \geq 1 + \int_1^\infty \frac{\fracc{2}v_n(\epsilon_N/\sqrt{z})^2}{2v_n\epsilon_N^2} = 1 + \fracc{4}\int_1^\infty z^{-1}\, dz
    \end{equation*}
    for sufficiently large $N$, and once again the integral diverges.
\end{proof}
Below, we generalize \Cref{lem:ratio_sample_mean_n3} to dimension $n = 2$ by using a truncation argument to circumvent the problem of infinite variance. Most of the details are in \Cref{sec:n2_details}.
\begin{lemma}\label{lem:ratio_sample_mean_all_n}
    For any $n \geq 2$, we have that $T_N/\deg(x_N) \xrightarrow[]{\p} \mu$ as $N \to \infty$.
\end{lemma}
\begin{proof}
    Let $\eta > 0$ and $\xi > 0$. We wish to show that
    \begin{equation*}
        \lim_{N \to \infty} \p[ |\ratiosumA/\deg(x_N) - \mu| > \eta] < \xi\,.
    \end{equation*}
    for sufficiently large $N$. We may expand $\p[ |\ratiosumA/\deg(x_N)  - \mu| > \eta$ as the sum
    \begin{align*}
        \p[ |\ratiosumA/\deg(x_N) - \mu_N | > \eta] &= \sum_{k= 0}^\infty \p[ |\ratiosumA/\deg(x_N) - \mu_N| > \eta \,\vert\, \deg(x_N) = k] \cdot \p[\deg(x_N) = k]\\
        &= \sum_{k = 0}^{\infty} \p[ |\Snk/k - \mu| > \eta] \cdot \p[\deg(x_N) = k]\,.
    \end{align*}
    By \Cref{lem:uniform_LLN_n2}, there is a constant $C$ and integers $N_*$, $k_*$ such that for $N \geq N_*$ and $k \geq k_*$,
    \begin{equation*}
        \p[ |\Snk/k - \mu| > \eta] < \xi + C \epsilon_N^2 \log(k)\,.
    \end{equation*}
    Therefore,
    \begin{align*}
         \p[ |\ratiosumA/\deg(x_N) - \mu_N | > \eta] &\leq \sum_{k \leq k_*} \p[ |\Snk/k - \mu_N| > \eta/2] \cdot \p[\deg(x_N) = k] \\
         &\qquad\qquad + \sum_{k > k_*} \p[ |\Snk/k - \mu_N| > \eta/2] \cdot \p[\deg(x_N) = k] \\
         &\leq \p[\deg(x_N) \leq k_*] + \sum_{k > k_*}\Big(\xi + C\epsilon_N^2 \log(k)\Big) \p[\deg(x_N) = k] \\
         &\leq \p[\deg(x_N) \leq k_*] + \xi \p[\deg(x_N) \geq k] + C\epsilon_N^2 \E[\log(\deg(x_N))] \\
         &\leq \p[\deg(x_N) \leq k_*] + \xi + C\epsilon_N^2 \E[\log(\deg(x_N))] \\
         &\leq \p[\deg(x_N) \leq k_*] + \xi + C\epsilon_N^2 \log\E[\deg(x_N)]\,,
    \end{align*}
    where the last inequality follows from applying Jensen's inequality. We have $\lim_{N \to \infty} \p[\deg(x_N) \leq k_*] = 0$ by \Cref{lem:degree_large}, so to finish the proof, it suffices to control $\epsilon_N^2 \log\E[\deg(x_N)]$ as $N \to \infty$. By construction, \[\E[\deg(x_N)] = (N-1) \frac{\vol(\BM{x_N}{\epsilon_N})}{\vol(M)} \propto N \epsilon_N^n,\] so
    \begin{equation*}
        \epsilon_N^2 \log(\E[\deg(x_N)]) = \mathcal{O}\Big(\epsilon_N^2 \log(N \epsilon_N^n)\Big) = \mathcal{O}\Bigg( \frac{\log(N^{1 - n\alpha})}{N^{2\alpha}}\Bigg)\,,
    \end{equation*}
    where $1 - n\alpha > 0$ by hypothesis. Therefore, $\lim_{N \to \infty} \epsilon_N^2 \log (\E[\deg(x_N)]) = 0$.
\end{proof}

\subsection{Proof that scalar Ollivier-Ricci curvature converges to scalar curvature}

\begin{lemma}\label{lem:orsc_rsc}
     % Suppose that $\orc{x_N}{y}$ is calculated with $\epsilon_N$-radius graph balls for all nodes $y \in G_N$ that are adjacent to $x_N$. 
     As $N \to \infty$,
    \begin{equation*}
        \fracc{\epsilon_N^4} \Big\vert \sorc{x_N} - \frac{\rsc(x_N)}{2(n+2)}\Big\vert \xrightarrow[]{\p} 0\,.
    \end{equation*}
\end{lemma}
\begin{proof}
    Assume that $N$ is sufficiently large so that $\epsilon_N$ is less than the injectivity radius of $M$. Therefore, $d_M(x_N, y)$ is less than the injectivity radius for all nodes $y$ that are adjacent to $x_N$. For $y \sim x_N$, let $v_y = \frac{\log_{x_N}(y)}{\norm{\log_{x_N}(y)}} = \frac{\log_{x_N}(y)}{d_M(x_N, y)}$. For any $x_N \in G_N$, we have
    \begin{align*}
        &\fracc{\epsilon_N^4}\Big| \orsc{x_N} - \frac{\rsc(x_N)}{2(n+2)}\Big| \\
        &\qquad\leq \deg(x_N)^{-1} \sum_{y \sim x} \fracc{\epsilon_N^2} \Big| \frac{d_M(x_N, y)^2}{\epsilon_N^2} \cdot \orc{x_N}{y} - \frac{\Ricc{x_N}{\log_{x_N}(y)}}{2(n+2)} \Big| \\
        &\qquad= \deg(x_N)^{-1} \sum_{y \sim x} \frac{d_M(x_N, y)^2}{\epsilon_N^2} \Big|\orc{x_N}{y} - \frac{\Ricc{x_N}{v_y}}{2(n+2)}\Big| \\
        &\qquad= \deg(x_N)^{-1} \sum_{y \sim x} \frac{d_M(x_N, y)^2}{\epsilon_N^2} \cdot \epsilon_N^{-2} \Big\vert \orc{x_N}{y} - \frac{\epsilon_N^2 \Ricc{x_N}{v_y}}{2(n+2)} \Big| \\
        &\qquad= \deg(x_N)^{-1}\sum_{y \sim x}  \epsilon_N^{-2}  \Big\vert \orc{x_N}{y}- \frac{\epsilon_N^2 \Ricc{x_N}{v_y}}{2(n+2)}\Big| \,.
    \end{align*}
    By comparing to $\kappa_M(x_N, y)$, we see that
    \begin{align*}
        &\fracc{\epsilon_N^4}\Big| \orsc{x_N} - \frac{\rsc(x_N)}{2(n+2)}\Big| \\
        &\qquad\leq \deg(x_N)^{-1} C^{-2} \sum_{y \sim x_N} \epsilon_N^{-2} 
    \Bigg( \vert \orc{x_N}{y} - \kappa_M(x_N, y)| + \Big\vert \kappa_M(x_N, y) - \frac{\epsilon_N^2 \Ricc{x_N}{v_y}}{2(n+2)}\Big| \Bigg)\\
        &\qquad= \deg(x_N)^{-1} C^{-2} \sum_{y \sim x_N} 
    \frac{|\W{M}(\mu_{x_N}^M, \mu_y^M) - \W{G}(\mu_{x_N}^G, \mu_y^G)|}{\epsilon_N^3}\cdot \frac{\epsilon_N}{d_M(x_N, y)} \\
    &\qquad\qquad\qquad + \deg(x_N)^{-1} C^{-2} \sum_{y \sim x_N} \epsilon_N^{-2}\Big\vert \kappa_M(x_N, y) - \frac{\epsilon_N^2 \Ricc{x_N}{v_y}}{2(n+2)}\Big|\,.
    \end{align*}
    
    By \Cref{thm:KM_convergence},
    \begin{equation*}
        \epsilon_N^{-2}\Big\vert \kappa_M(x_N, y) - \frac{\epsilon_N^2 \Ricc{x_N}{v_y}}{2(n+2)} \Big| = \mathcal{O}\Big(\epsilon_N + d_M(x_N, y) \Big) = \mathcal{O}(\epsilon_N)
    \end{equation*}
    for any points $x_N$ and $y$. Therefore, it suffices to show that
    \begin{equation*}
        \deg(x_N)^{-1}\sum_{y \sim x_N} 
    \frac{|\W{M}(\mu_{x_N}^M, \mu_y^M) - \W{G}(\mu_{x_N}^G, \mu_y^G)|}{\epsilon_N^3}\cdot \frac{\epsilon_N}{d_M(x_N, y)} \to 0
    \end{equation*}
    in probability as $N \to \infty$. We upper bound the sum by
    \begin{align*}
        &\deg(x_N)^{-1}\sum_{y \sim x_N} 
    \frac{|\W{M}(\mu_{x_N}^M, \mu_y^M) - \W{G}(\mu_{x_N}^G, \mu_y^G)|}{\epsilon_N^3}\cdot \frac{\epsilon_N}{d_M(x_N, y)} \\
    &\qquad\qquad \leq \Bigg(\max_{y \sim x_N}  \frac{|\W{M}(\mu_{x_N}^M, \mu_y^M) - \W{G}(\mu_{x_N}^G, \mu_y^G)|}{\epsilon_N^3}\Bigg) \cdot \deg(x_N)^{-1} \sum_{y \sim x_N} \frac{\epsilon_N}{d_M(x_N, y)}\,.
    \end{align*}
    Applying \Cref{lem:WM_convergence_modified} and \Cref{lem:ratio_sample_mean_all_n} completes the proof.
\end{proof}

\begin{theorem}\label{thm:main}
We have
    \begin{equation*}
        \Big|\frac{2(n+2)^2}{\epsilon_N^4} \sorc{x_N} - S(x_N)\Big| \xrightarrow[]{\p} 0
    \end{equation*}
as $N \to \infty$.
\end{theorem}
\begin{proof}
    By \Cref{lem:orsc_rsc},
    \begin{equation*}
        \Big| \frac{2(n+2)^2}{\epsilon_N^4} \sorc{x_N} - \frac{(n+2)\rsc(x_N)}{\epsilon_N^2} \Big| \to 0
    \end{equation*}
    in probability as $N \to \infty$. By \Cref{lem:rsc_convergence},
    \begin{equation*}
        \Big| \frac{(n+2)\rsc(x_N)}{\epsilon_N^2}- S(x_N)\Big| \to 0
    \end{equation*}
    in probability as $N \to \infty$.
\end{proof}

\section{Numerical Experiments}
We calculate our scalar ORC for all nodes in a set of random geometric graphs that are sampled from unit spheres of dimensions $n \in \{2, 3, 4\}$. Our spheres have constant scalar curvatures $S = 2, 6, 12$, respectively. According to \Cref{thm:main}, $\frac{2(n+2)^2}{\epsilon_N^4}\sorc{x}$ converges to the scalar curvature $S$ of the sphere as the number of nodes $N \to \infty$, where $\epsilon_N$ is the connectivity threshold. In \Cref{fig:sphere_convergence}, we plot the convergence of the mean scaled SORC (averaged over all nodes in the graph) as the number $N$ of nodes increases.

\begin{figure}
    \centering
    \subfloat[Sparser RGGs]{\includegraphics[width = .45\textwidth]{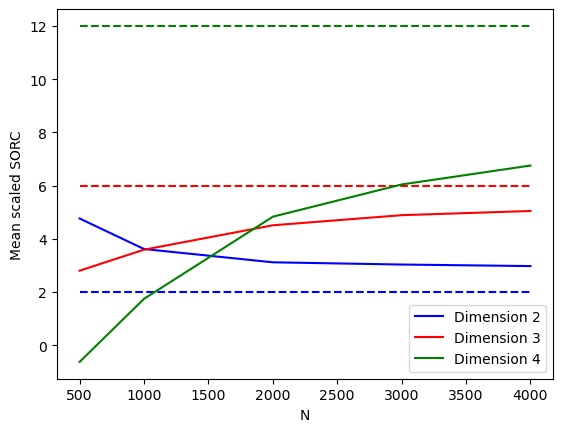}}
    \hspace{5mm}
    \subfloat[Denser RGGs]{\includegraphics[width = .45\textwidth]{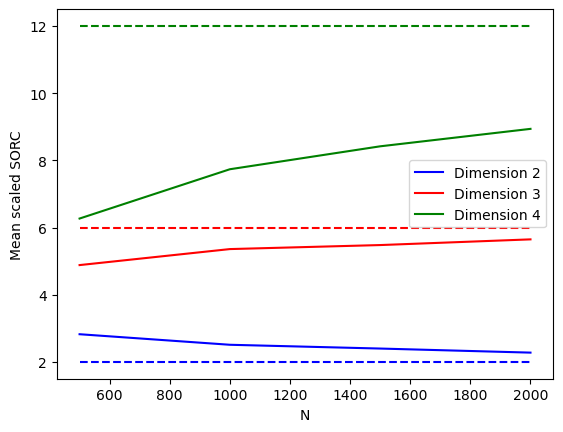}}
    \caption{The graphs are RGGs with nodes sampled from spheres of dimensions $n \in \{2, 3, 4\}$. The solid lines are the mean $\frac{2(n+2)^2}{\epsilon_N^4}\sorc{x}$ as a function of the number $N$ of nodes in the graph. The dashed lines are the scalar curvatures of the corresponding spheres. The connectivity threshold $\epsilon_N$ is set to $C_n N^{-\alpha_n}$, where $\alpha_n = \fracc{6.01n}$ so that $\alpha$ satisfies our condition $\alpha \in (0, \fracc{6n})$. In (A), the constant $C_n$ is set so that the average degree at $N = 1000$ is approximately $50$. In (B), the constant $C_n$ is set so that the average degree at $N = 1000$ is approximately $100$. More precisely, if $k$ is the desired average degree at $N = 1000$, we set $C_n$ such that $k = N \frac{v_n \epsilon_N^n}{\vol(S^n)}$, which is approximately equal to the expected degree of a node, ignoring curvature.}
    \label{fig:sphere_convergence}
\end{figure}

\section{Conclusions}

In this paper, we defined ``scalar Ollivier-Ricci curvature,'' a function $\sorc{x}$ on the nodes of a weighted graph. The scalar Ollivier-Ricci curvature at a node $x$ is a weighted sum of the Ollivier-Ricci curvature of each incident edge, where higher-weight edges are given higher weight in the $\sorc{x}$ sum. We proved (\Cref{thm:main}) that $\sorc{x}$ converges in probability to a scaled version of scalar curvature, as the  number of nodes $N \to \infty$ in a random geometric graph.

Somewhat counterintuitively, our definition of $\sorc{x}$ places higher weight on $\orc{x}{y}$ when $x$ and $y$ are farther apart (but still connected) in a random geometric graph. This reflects what we see in \Cref{prop:orc_convergence1} (and also \cite{ORC_convergence_hoorn, trillos_weber}); when $x$ and $y$ are too close together, the error grows between Ollivier-Ricci curvature and Ricci curvature. If $x$ and $y$ are very close, then their $1$-hop neighborhoods have significant overlap, so their edge is likely to have positive Ollivier-Ricci curvature $\kappa_G(x, y)$, regardless of the true underlying curvature of the manifold. On the other hand, if $x$ and $y$ are a little farther apart, but still close enough to be connected, then curvature plays a larger role in the connectivity of their respective neighborhoods, so that Ollivier-Ricci curvature is a good approximation to Ricci curvature.

Our theoretical analysis assumes a framework that is not always true in real-world networks, so we suggest the following pre-processing steps when our assumptions are violated. First, our theoretical analysis assumes that edge weights are geodesic distances, so that edges with lower weight correspond to points that are closer to each other. However, it is often the opposite case in real-world networks where higher weight indicates a stronger connection. In those cases, we suggest applying a monotonically decreasing function $f(w)$ to the weights $w$ in order to obtain weights $\tilde{w}$ in which lower weight indicates a closer connection. Second, our theoretical analysis assumes that there is a connectivity threshold $\epsilon$ at which points $x, y$ are not connected by an edge if their geodesic distance is greater than $\epsilon$. In real-world networks, one should apply a threshold so that each node is only connected to other nodes that have a reasonably strong connection with it---or to nodes that are reasonably close to it, depending on context. Otherwise, our definition of $\sorc{x}$ will put very high weight on edges that are not truly relevant to $x$. With these two pre-processing steps, one can sensibly calculate scalar Ollivier-Ricci curvature on real-world networks that are not necessarily random geometric graphs.
\appendix

\section{Technical Details}\label{sec:n2_details}
Define the random variables $Z_N^i, Z$ and their respective means $\mu_N, \mu$ as in \Cref{sec:ratio}. For every integer $k$, define the truncated random variables
\begin{align*}
    \overline{Z_{N, k}^i} := Z_N^i \mathbf{1}_{|Z_N^i| \leq k}\,, \\
    \Ztrunc := Z \mathbf{1}_{|Z| \leq k}
\end{align*}
and their respective means
\begin{align*}
    \mu_{N, k} &:= \E[Z_{N, k}^1]\,,\\
    \mu_k &:=  \E[\Ztrunc]\,.
\end{align*}
Define the sum $\Snk = \sum_{i=1}^k Z_N^i$ as in \Cref{sec:ratio}, and define the analogous truncation sum  
\begin{equation*}
    \Snktrunc = \sum_{i=1}^k \overline{Z_{N, k}^i}\,.
\end{equation*}

\begin{lemma}\label{lem:truncated_mean_convergence}
Let $\eta > 0$. There are $N_* \in \mathbb{N}$ and $k_* \in \mathbb{N}$ such that for $N \geq N_*$ and $k \geq k_*$, we have
    \begin{equation*}
        |\mu_{N, k} - \mu| < \eta\,.
    \end{equation*}
\end{lemma}
\begin{proof}
    For any $N$ and $k$,
    \begin{equation*}
        |\mu_{N, k} - \mu_N| \leq |\mu_{N, k} - \mu_k| + |\mu_k - \mu|\,.
    \end{equation*}
    The second term corresponds to convergence in the Euclidean case as the truncation parameter $k \to \infty$. There is no $N$ dependence. Because $Z$ is a non-negative random variable,
    \begin{align*}
        \mu &:= \int_0^{\infty} \p[Z > z] \,dz \\
        \mu_k &:= \int_0^{\infty} \p[Z \mathbf{1}_{Z \leq k} > z] \,dz\,.
    \end{align*}
    We note that $\lim_{k \to \infty} \p[Z \mathbf{1}_{Z \leq k} > z] = \p[Z > z]$, and the latter is an upper bound at each $k$. By the dominated convergence theorem, $\vert \mu_k - \mu\vert \to 0$ as $k \to \infty$.

    We now turn to the first term, $|\mu_{N, k} - \mu_N|$. Because $\Zntrunc$ and $\Ztrunc$ are non-negative random variables, we have
    \begin{align*}
        \mu_{N, k} &:= \int_0^{\infty} \p[\Zntrunc > z]\,dz\,, \\
        \mu_k &:= \int_0^{\infty} \p[\Ztrunc > z]\,dz\,.
    \end{align*}
    Therefore,
    \begin{equation*}
        \mu_{N, k} - \mu_k = \int_0^{\infty} \p[\Ztrunc \leq z] - \p[\Zntrunc \leq z]\, \,dz.
    \end{equation*}
    For any non-negative random variable $Y$ and its truncation $\overline{Y_k} := Y \mathbf{1}_{|Y|\leq k}$, we have
    \begin{equation*}
        \p[\overline{Y_k} \leq y] = \begin{cases}
            \p[Y \leq k] - \p[Y \leq y] \,, & y \leq k \,, \\
            1 \,, & \text{otherwise}\,.
        \end{cases}
    \end{equation*}
    Therefore,
    \begin{equation*}
        |\mu_{N, k} - \mu_k| = \int_0^k |\p[Z\leq k] - \p[Z_N^1 \leq k]| + |\p[Z \leq z] - \p[Z_N^1\leq z] | \,dz \,.
    \end{equation*}
    Because $Z_N^1 \geq 1$ and $Z\geq 1$, 
    \begin{equation*}
        |\mu_{N, k} - \mu_k| = \int_1^k |\p[Z\leq k] - \p[Z_N^1 \leq k]| + |\p[Z \leq z] - \p[Z_N^1\leq z] | \,dz\,.
    \end{equation*}
    By \Cref{lem:ratio_convergence_dist}, there is a constant $C$ such that for sufficiently large $N$,
    \begin{equation*}
        |\p[Z \leq z] - \p[Z_N^1\leq z] | \leq C \epsilon_N^2/z^n
    \end{equation*}
    for all $z \geq 1$ (including $z = k$). Therefore,
    \begin{equation*}
        |\mu_{N, k} - \mu_k| \leq C\epsilon_N^2 \Big( k^{-n+1} + \int_1^k z^{-n} \,dz\Big) = \mathcal{O}(\epsilon_N^2 k^{-n+1})\,,
    \end{equation*}
    which implies the lemma.
\end{proof}

\begin{lemma}\label{lem:kprobZ}
    \begin{equation*}
        \lim_{z \to \infty} z \p[Z > z] = 0\,.
    \end{equation*}
\end{lemma}
\begin{proof}
    This is a standard proof (see e.g., Theorem 2.2.14 in \cite{durrett_probability}), which we include here for convenience. For any $z$,
    \begin{equation*}
        z\p[Z > z] \leq \E[Z \mathbf{1}_{Z > z}] \leq \E[Z \mathbf{1}_{Z > \lfloor z \rfloor}]\,.
    \end{equation*}
    Therefore, it suffices to show that $\lim_{m \to \infty} \E[Z \mathbf{1}_{Z >m}] = 0$, where $m \in \mathbb{Z}$. Let $f_m(z) = \p[Z \mathbf{1}_{Z > m} > z]$ and let $g(z) = \p[Z > z]$. Because $Z$ is non-negative, $\E[Z \mathbf{1}_{Z >m}] = \int_0^{\infty} f_m(z)\,dz$. The function $g(z)$ is an upper bound for $f_m(z)$, and it is integrable because $\int_0^{\infty} \,g(z) \,dz = \E[Z] < \infty$. Additionally, $\lim_{m \to \infty} f_m(z) \leq \lim_{m \to \infty} \p[Z > m] = 0$, so by the dominated convergence theorem, $\lim_{m \to \infty} \E[Z \mathbf{1}_{Z >m}] = 0$.
\end{proof}
\begin{lemma}\label{lem:k_truncated}
Let $\eta > 0$. There are $N_* \in \mathbb{N}$ and $k_* \in \mathbb{N}$ such that for $N \geq N_*$ and $k \geq k_*$, we have
    \begin{equation*}
        k \p[|Z_N^1| > k] < \eta\,.
    \end{equation*}
\end{lemma}
\begin{proof}
    For all $N$ and $k$,
    \begin{equation*}
        k \p[|Z_N^1| > k] \leq k|\p[Z_N^1 > k] - \p[Z > k]| + k\p[Z > k]\,.
    \end{equation*}
    By \Cref{lem:kprobZ}, $\lim_{k \to \infty} k \p[Z > k] = 0$, so it suffices to control $k|\p[Z_N^1 > k] - \p[Z > k]|$. By \Cref{lem:ratio_convergence_dist}, there is a constant $C$ such that
    \begin{equation*}
        k|\p[Z_N^1 > k] - \p[Z > k]| \leq C\cdot \epsilon_N^2 k^{-n +1}\,,
    \end{equation*}
    which is $\mathcal{O}(\epsilon_N^2 k^{-1})$ for $n \geq 2$.
\end{proof}

\begin{lemma}\label{lem:truncated_variance}
    Let $\xi > 0$. Then there is a constant $C$, an integer $N_* \in \mathbb{N}$, and an integer $k_* \in \mathbb{N}$ such that for $N \geq N_*$ and $k \geq k_*$, we have
    \begin{equation*}
        k^{-1} \E[ (\Zntrunc)^2] < \xi + C \epsilon_N^2 \log(k) \,.
    \end{equation*}
\end{lemma}
\begin{proof}
    By Lemma 2.2.13 in \cite{durrett_probability}, we have that
    \begin{equation*}
        k^{-1}\E[ (\Zntrunc)^2] = k^{-1}\int_0^{\infty} 2 z \p[ \Zntrunc > z] \,dz\,.
    \end{equation*}
    Because $\Zntrunc$ is truncated,
    \begin{equation*}
        k^{-1}\int_0^{\infty} 2 z \p[ \Zntrunc > z] \,dz = k^{-1}\int_0^k 2 z \p[ \Zntrunc > z] \,dz \leq k^{-1}\int_0^k 2z \p[Z_N^1 > z] \,dz\,.
    \end{equation*}
    We note that
    \begin{equation*}
        k^{-1} \int_0^k 2z \p[Z_N^1 > z] \,dz \leq k^{-1} \int_0^k 2z|\p[Z_N^1 > z] - \p[Z > z]| \,dz + k^{-1} \int_0^k 2z\p[Z > z] \,dz\,.
    \end{equation*}
    Putting the equations above together,
    \begin{equation}\label{eq:truncated_var_1}
         k^{-1}\E[ (\Zntrunc)^2] \leq k^{-1} \int_0^k 2z|\p[Z_N^1 > z] - \p[Z > z]|dz + k^{-1} \int_0^k 2z\p[Z > z]dz\,.
    \end{equation}
When $z \in [0, 1]$, we have $\p[Z_N^1 > z] = \p[Z > z] = 1$ because $Z_N^1 \geq 1$ and $Z \geq 1$. By \Cref{lem:ratio_convergence_dist}, there is a constant $C$ such that
 \begin{equation*}
     |\p[Z_N^1 > z] - \p[Z > z]| \leq C\epsilon_N^2 z^{-n}
 \end{equation*}
 for sufficiently large $N$ and all $z \geq 1$. Therefore,
 \begin{equation*}
     k^{-1} \int_0^k 2z|\p[Z_N^1 > z] - \p[Z > z]|\,dz = k^{-1} \int_1^k 2z|\p[Z_N^1 > z] - \p[Z > z]|\,dz \leq 2C \epsilon_N^2 \int_1^k z^{-n+1} \,dz\,.
 \end{equation*}
The worst-case scenario is $n = 2$, in which case the right-hand side is
 \begin{equation*}
     2C \epsilon_N^2 \log(k)\,,
 \end{equation*}
 so by \Cref{eq:truncated_var_1},
 \begin{equation*}
      k^{-1}\E[ (\Zntrunc)^2] \leq 2C \epsilon_N^2 \log(k) + k^{-1} \int_0^k 2z \p[Z > z]\,dz\,.
 \end{equation*}
    
    By a change of variables to the integral on the right-hand side,
    \begin{equation*}
        k^{-1}\int_0^k 2 z \p[ Z > z] \,dz = \int_0^1 2kz \p[Z > kz] \,dz\,.
    \end{equation*}
    The right-hand side is independent of $N$ and goes to $0$ as $k \to \infty$ by the dominated convergence theorem, so it is bounded above by $\xi$ for sufficiently large $k$ and all $N$.
 
\end{proof}

\begin{lemma}\label{lem:uniform_LLN_n2}
    Let $\eta > 0$ and $\xi > 0$. Then there is a constant $C$, an integer $N_* \in \mathbb{N}$, and an integer $k_* \in \mathbb{N}$ such that for $N \geq N_*$ and $k \geq k_*$,
    \begin{equation*}
        \p[ |S_{N, k}/k - \mu| > \eta] < \xi + C \epsilon_N^2 \log(k)\,.
    \end{equation*}
\end{lemma}
\begin{proof}
    By \Cref{lem:truncated_mean_convergence}, there are $N_1$ and $k_1$ such that if $N \geq N_1$ and $k \geq k_1$, then
    \begin{equation*}
        |\mu_{N, k} - \mu| < \eta / 2\,,
    \end{equation*}
    so
    \begin{equation*}
        \p[ |S_{N, k}/k - \mu| > \eta] < \p[ |S_{N, k}/k - \mu_{N, k}| > \eta/2]\,.
    \end{equation*}
    Then we have
    \begin{equation*}
        \p[ |S_{N, k}/k - \mu_{N, k}| > \eta/2] \leq \p[S_{N, k} \neq \Snktrunc] + \p[ |\Snktrunc/k - \mu_{N, k} | > \eta/2]\,.
    \end{equation*}
    We can upper bound $\p[S_{N, k} \neq \Snktrunc]$ by observing that $S_{N, k} \neq \Snktrunc$ only if $|Z_N^i| > k$ for some $1 \leq i \leq k$. By applying a union bound, we have that 
    $\p[S_{N, k} \neq \Snktrunc] \leq k \p[|Z_N^1| > k]$. Therefore, by \Cref{lem:k_truncated}, there are $N_2 > N_1$ and $k_2 > k_1$ such that
    \begin{equation*}
        \p[S_{N, k} \neq \Snktrunc] < \xi/2
    \end{equation*}
    for all $N > N_2$ and $k > k_2$. 
    
    By Chebyschev's inequality and straightforward simplification,
    \begin{align*}
        \p[ |\Snktrunc/k - \mu_{N, k} | > \eta/2] 
        &\leq 4\eta^{-2} \var(\Snktrunc/k) \\
        &= 4\eta^{-2}k^{-1} \var(\Zntrunc) \\
        &\leq  4\eta^{-2}k^{-1} \E[ (\Zntrunc)^2]\,.
    \end{align*}
    By \Cref{lem:truncated_variance}, there is a constant $C$ and integers $N_3 \geq N_2$, $k_3 \geq k_2$ such that 
    \begin{equation*}
        4\eta^{-2}k^{-1} \E[ (\Zntrunc)^2] < \xi/2 + C \epsilon_N^2 \log(k) \,.
    \end{equation*}
    for all $N \geq N_3$ and $k \geq k_3$, which concludes the proof.
\end{proof}
\bibliographystyle{plain}
\bibliography{references}
\end{document}